\documentclass[12pt]{article}
\pdfoutput=1
\usepackage{makeidx}
\makeindex
\usepackage[a4paper]{geometry}
\usepackage{jheppub,amsmath,amssymb,amsfonts,amsxtra, mathrsfs,graphics,graphicx,amsthm,epsfig,ytableau,bm,longtable,float,color,tikz,mathtools,xfrac,footnote,rotating,lscape}
\usepackage{stackengine,scalerel}
\stackMath

\usepackage[all]{xy}
\usetikzlibrary{decorations.pathmorphing}
\usetikzlibrary{decorations.markings}
\usetikzlibrary{arrows, decorations.markings, calc, fadings, decorations.pathreplacing, patterns, decorations.pathmorphing, positioning}

\usetikzlibrary{positioning,shapes}
\usetikzlibrary{chains}
\usetikzlibrary{arrows,fit,decorations.pathreplacing}
\tikzstyle{every picture}+=[remember picture]
\tikzstyle{na} = [baseline=-.5ex]

\usepackage{empheq}
\usepackage{multirow}
\usepackage{booktabs}
\usepackage[american]{babel}

\usepackage[latin1]{inputenc}

\makeatletter
\newcommand{\vast}{\bBigg@{1}}
\newcommand{\Vast}{\bBigg@{5}}
\makeatother

\usepackage{hyperref}

\numberwithin{equation}{section}

\newcommand{\ie}{\textit{i.e.}}

\numberwithin{equation}{section}

\newcommand{\nn}{\nonumber}

\newcommand{\be}{\begin{eqnarray}}\newcommand{\ee}{\end{eqnarray}}
\newcommand{\bea}{\begin{equation} \begin{aligned}} \newcommand{\eea}{\end{aligned} \end{equation}}

\usepackage{dsfont}

\def\U{\mathrm{U}}
\def\SU{\mathrm{SU}}
\def\SL{\mathrm{SL}}

\newcommand{\wb}{\overline}
\newcommand{\wt}{\widetilde}

\DeclareMathOperator{\Tr}{Tr}

\DeclareMathOperator{\sign}{sign}

\DeclareMathOperator{\re}{\mathbb{R}e}
\DeclareMathOperator{\im}{\mathbb{I}m}

\newcommand{\cA}{\mathcal{A}}
\newcommand{\cB}{\mathcal{B}}
\newcommand{\cC}{\mathcal{C}}

\newcommand{\cN}{\mathcal{N}}
\newcommand{\cO}{\mathcal{O}}
\newcommand{\cP}{\mathcal{P}}

\newcommand{\cV}{\mathcal{V}}
\newcommand{\cW}{\mathcal{W}}

\newcommand{\bZ}{\mathbb{Z}}

\newcommand{\fh}{\mathfrak{h}}
\newcommand{\fm}{\mathfrak{m}}

\newcommand{\fn}{\mathfrak{n}}
\newcommand{\fp}{\mathfrak{p}}

\newcommand{\fR}{\mathfrak{R}}

\newcommand{\tv}{\tilde{v}}
\newcommand{\tu}{\tilde{u}}
\newcommand{\tx}{\tilde{x}}

\usepackage[Symbol]{upgreek}
\usepackage{bm}

\usepackage{calligra}
\DeclareMathAlphabet{\mathcalligra}{T1}{calligra}{m}{n}

\newtheorem{theorem}{Theorem}

\title{The Cardy limit of the topologically twisted index and black strings in AdS$_5$}

\author{Seyed Morteza Hosseini,}
\author{Anton Nedelin}
\author{and Alberto Zaffaroni}

\affiliation{Dipartimento di Fisica, Universit\`a di Milano-Bicocca, I-20126 Milano, Italy}
\affiliation{INFN, sezione di Milano-Bicocca, I-20126 Milano, Italy}

\emailAdd{morteza.hosseini@mib.infn.it}
\emailAdd{anton.nedelin@unimib.it}
\emailAdd{alberto.zaffaroni@mib.infn.it}

\abstract{We evaluate the topologically twisted index of a general four-dimensional $\mathcal{N} = 1$ gauge theory in the ``high-temperature" limit. The index is the partition function for $\mathcal{N} = 1$ theories on $S^2 \times T^2$, with a partial topological twist along $S^2$, in the presence of background magnetic fluxes  and fugacities for the global symmetries. We show that the logarithm of the index 
is proportional to the conformal anomaly coefficient of the two-dimensional $\mathcal{N} = (0,2)$ SCFTs obtained from the compactification  on $S^2$. We also present a universal formula  for extracting the index from the four-dimensional conformal anomaly coefficient and its derivatives. We give examples based on theories whose holographic duals are black strings in type IIB backgrounds AdS$_5 \times \text{SE}_5$, where SE$_5$ are five-dimensional Sasaki-Einstein spaces.
}

\begin{document}

\setcounter{tocdepth}{3}
\maketitle

%
%

\date{Dated: \today}




\section{Introduction} 

The topologically twisted index introduced in \cite{Benini:2015noa} is the partition function for three- and four-dimensional gauge theories with at least four supercharges on $\Sigma_g\times T^d$, where $d=1,2$, with a topological $A$-twist on $\Sigma_g$.
When it is refined with chemical potentials and background magnetic charges for the flavor symmetries, it becomes an efficient tool for studying the nonperturbative properties of supersymmetric gauge theories \cite{Closset:2013sxa,Benini:2015noa,Honda:2015yha,Closset:2015rna,Benini:2016hjo,Closset:2016arn}. The large $N$ limit of the index contains interesting information about theories with a holographic dual. In particular, the large $N$ limit of the index  for  the three-dimensional ABJM theory was successfully used in  \cite{Benini:2015eyy,Benini:2016rke} to provide the first microscopic counting of the microstates
of an AdS$_4$ black hole. The large $N$ limit of general three-dimensional quivers with an AdS dual was studied in \cite{Hosseini:2016tor,Hosseini:2016ume}. In this paper
we study the asymptotic  behavior of the index, at finite $N$, for four-dimensional $\cN =1$ gauge theories.

With an eye on holography we also evaluate the index in the large $N$ limit.
We focus, in particular, on the class of $\cN =1$ theories  arising from D3-branes probing Calabi-Yau singularities, which have a well-known holographic dual in terms of
compactifications on Sasaki-Einstein manifolds. Black string solutions corresponding to D3-branes at a Calabi-Yau singularity 
have been recently studied in details  in \cite{Benini:2012cz,Benini:2013cda,Benini:2015bwz}.
They  interpolate between AdS$_5$ and AdS$_3\times \Sigma_g$ vacua and can be interpreted as an RG flow from an UV four-dimensional $\cN=1$ CFT  and an IR two-dimensional $(0,2)$ one.
The two-dimensional CFT is obtained by compactifying the four-dimensional theory on $\Sigma_g$  with a topological twist parameterized by a set of background magnetic charges $\fn_I$.
The right-moving central charge of the two-dimensional CFT has  been computed  in  \cite{Benini:2012cz,Benini:2013cda,Bobev:2014jva,Benini:2015bwz}, and successfully compared with the supergravity result for a variety of models.  

The topologically twisted index of a general four-dimensional $\cN =1$ gauge theory can be interpreted as a trace over a Hilbert space of states on $\Sigma_g \times S^1$
\be
\label{trace}
Z (\fn, y) =  {\rm Tr}_{\Sigma_g \times S^1}  (-1)^F q^{H_L} \prod_I y_I^{J_I} \, ,
\ee
where $q=e^{2\pi i\tau}$ and $y_I$ are fugacities for the flavor symmetries $J_I$.
Here, $\tau$ is the complex modulus of $T^2$.
The Hamiltonian $H_L$ on $\Sigma_g \times S^1$ explicitly depends on the background magnetic fluxes $\fn_I$.
For simplicity, we restrict to the case of $\Sigma_g=S^2$, since the generalization to an arbitrary Riemann surface is straightforward \cite{Benini:2016hjo}.
The index can be evaluated using supersymmetric  localization and it reduces to a matrix model. It can be written as the contour integral, 
\begin{equation}
Z (\fn, y) = \frac1{|\cW|} \; \sum_{\fm \,\in\, \Gamma_\fh} \; \oint_\cC Z_{\text{int}} (\fm, x;  \fn , y) \, ,
\end{equation}
of a meromorphic differential form  in variables $x$ living on the torus $T^2$ and parameterizing the Cartan subgroup of the gauge group. 
An important feature of the matrix model is that there is a sum over
the lattice of magnetic charges $\fm$ of the gauge group.  For each $\fm$ the integrand has the form of an elliptic genus as computed in  \cite{Benini:2013xpa,Benini:2013nda}.
There exist particular choices of background magnetic fluxes $\fn$ for which the sum  truncates to a single set of  gauge fluxes $\fm$ \cite{Gadde:2015wta}.  
However, for generic background fluxes this does not happen and we need to sum an infinite number of contributions. The strategy is then  to explicitly resum the integrand \cite{Benini:2015eyy} and consider the contour integral of  
\begin{equation}
Z_{\text{resummed}}  ( x;  \fn , y) = \frac1{|\cW|} \; \sum_{\fm \,\in\, \Gamma_\fh}  Z_{\text{int}} (\fm, x;  \fn , y) \, ,
\end{equation}
which is a complicated elliptic function of $x$. One can write a set of algebraic equations for the position of the poles,
which we call {\it Bethe ansatz equations} (BAEs) (they actually are the BAEs of the dimensionally reduced theory on $\Sigma_g$ in the formalism of \cite{Nekrasov:2014xaa}),
and a {\it Bethe potential}   $\cV$ (or Yang-Yang functional \cite{Yang1969}) whose derivatives reproduce the BAEs. The topologically twisted index is then given  by
the sum of the residues of $Z_{\text{resummed}}$ at the solutions to the BAEs. 
The explicit evaluation of the topologically twisted index is a hard task, even in the large $N$ limit.
However, the index greatly simplifies if we identify  the modulus $\tau=i \beta/2\pi$ of the torus $T^2$ with a \emph{fictitious} inverse temperature $\beta$, and take the \emph{high-temperature} limit $(\beta \to 0)$. In this limit, we can use the modular properties of the integrand under the $\SL(2,\bZ)$ action to simplify the result.

In the high-temperature limit, we find a number of  interesting results, valid to leading order in $1 / \beta$.

First, we obtain an explicit relation between the Bethe potential and the R-symmetry 't\,Hooft anomalies of the UV four-dimensional $\cN=1$ theory
\bea\label{tHoof:anomaly:0}
\overline\cV \left( \Delta_I \right)= \frac{\pi^3}{6 \beta} \left[ \Tr R^3 ( \Delta_I) - \Tr R ( \Delta_I) \right]
= \frac{16 \pi^3}{27 \beta} \left[ 3c \left( \Delta_I \right) - 2 a ( \Delta_I) \right] \, ,
\eea
where $R$ is a choice of $\U(1)_R$ symmetry and the trace is over all fermions in the theory.
Here, we use the chemical potentials $\Delta_I/\pi$ to parameterize a trial R-symmetry of the $\cN=1$ theory.
Details about this identification are given in the main text.
In writing the second equality in \eqref{tHoof:anomaly:0} we used the relation between conformal and R-symmetry 't\,Hooft anomalies in $\cN = 1$ SCFTs \cite{Anselmi:1997am},
\be\label{generalac}
a = \frac{9}{32} \Tr R^3 - \frac{3}{32} \Tr R \, ,
\qquad \qquad c = \frac{9}{32} \Tr R^3 - \frac{5}{32} \Tr R \, .
\ee

Secondly, the value of the index as a function of the chemical potentials $\Delta_I$ and the set of magnetic fluxes $\fn_I$, parameterizing the twist,
can be expressed in terms of the trial left-moving central charge of the 2d $\cN = (0,2)$ SCFT  as
\be
 \label{index theorem:2d central charge0}
 \log Z (\Delta_I,\fn_I) = \frac{\pi^2}{6 \beta} c_{l} \left( \Delta_I , \fn_I \right) \, .
\ee
This is related to the trial right-moving central charge $c_r$ by the gravitational anomaly $k$ \cite{Benini:2012cz,Benini:2013cda},
\be
\label{cl:cr:k}
 c_r - c_l = k \, , \qquad \qquad k = - \Tr \gamma_3 \, .
\ee
Here, $\gamma_3$ is the chirality operator in two dimensions.%
\footnote{With our choice of chirality operator the gaugino zero-modes have $\gamma_3 = 1$.}

Finally, there is a simple universal formula at leading order in $N$ for computing the index from the Bethe potential as a function of the chemical potentials $\Delta_I$,
\bea
\label{Z large N conjecture0}
 \log Z(\Delta_I,\fn_I) = - \frac{3}{\pi} \, \overline\cV \left( \Delta_I \right) - \sum_{I} \left[ \left( \fn_I - \frac{\Delta_I}{\pi} \right) \frac{\partial \overline\cV \left( \Delta_I \right)}{\partial \Delta_I} \right]
 = \frac{\pi^2}{6 \beta} c_r \left( \Delta_I , \fn_I \right) \, ,  
\eea
where the index $I$ runs over the bi-fundamental and adjoint fields in the quiver.
In the large $N$ limit the Bethe potential can be written as
\bea\label{centralchargea0}
 \overline\cV \left( \Delta_I \right)= \frac{16 \pi^3}{27 \beta} a ( \Delta_I) \, .
\eea
These formulae are valid for theories of D3-branes, where $\Tr R=\cO(1)$ and $c=a$ at large $N$ \cite{Henningson:1998gx}. These topologically twisted 
theories have  holographic duals in terms of black strings in AdS$_5 \times \text{SE}_5$, where SE$_5$ are five-dimensional Sasaki-Einstein spaces \cite{Benini:2012cz,Benini:2013cda}.

There is a striking similarity with the results obtained in \cite{Benini:2015eyy,Benini:2016rke,Hosseini:2016tor,Hosseini:2016ume} for the large $N$ limit of the topologically twisted index of three-dimensional theories, if we replace
\begin{equation*}
 \begin{array}{ccc}
  \emph{\text{central charge }} a( \Delta_I) \; &\Longleftrightarrow & \; \emph{\text{free energy on }} S^3 \\
 \emph{\text{central charge }} c_{r} \left( \Delta_I , \fn_I \right) \; &  \Longleftrightarrow & \; \emph{\text{black hole entropy}} \\
 c-\emph{\text{extremization}} \; & \Longleftrightarrow & \; I-\emph{\text{extremization}} \, .
\end{array}
\end{equation*}
Indeed, in  three dimensions,  the very same formula \eqref{Z large N conjecture0} holds  with the Bethe potential given by the  $S^3$ partition function $F_{S^3}$ of the gauge theory \cite{Hosseini:2016tor,Hosseini:2016ume}. Notice that  $F_{S^3}$ is the natural replacement for $a$, both being monotonic along RG flows \cite{Intriligator:2003jj,Jafferis:2011zi}. Moreover, both of them can be computed, as a function of $\Delta_I$, in terms of the volume of a family of Sasakian manifolds \cite{Gubser:1998vd,Butti:2005vn,Martelli:2005tp,Martelli:2006yb,Jafferis:2011zi}. 
 In addition, in three dimensions, the dual black string is replaced by a dual black hole  and
$\log Z$ computes the entropy of the black hole. As discussed in \cite{Benini:2015eyy,Benini:2016rke,Benini:2015bwz}, the entropy is obtained by extremizing $\log Z$  with respect to the $\Delta_I$ ($I$-extremization). Similarly, as it was shown in  \cite{Benini:2012cz,Benini:2013cda}, the exact central charge of the 2d SCFT is obtained by extremizing the trial right-moving central charge with respect to the $\Delta_I$.
Given the relation \eqref{Z large N conjecture0} we see that $c$-extremization corresponds to  $I$-extremization.  
Finally, in both three and four dimensions, the field theory extremization corresponds to the attractor mechanism \cite{Ferrara:1996dd, Dall'Agata:2010gj,Karndumri:2013iqa,Hristov:2014eza,Amariti:2016mnz,Klemm:2016kxw} on the gravity side.

Formula \eqref{index theorem:2d central charge0} implies a Cardy-like behavior of the topologically twisted index, which is related to the modular properties of the elliptic genus \cite{Kawai:1993jk,Benjamin:2015hsa}.
Analogous behaviors for other partition functions have been found in \cite{DiPietro:2014bca,Ardehali:2015hya,Ardehali:2015bla,Lorenzen:2014pna,Assel:2015nca,Bobev:2015kza,Genolini:2016sxe,DiPietro:2016ond,Brunner:2016nyk,Shaghoulian:2015kta,Shaghoulian:2015lcn}.\footnote{In particular, ambiguities in the definition of the partition function for 2d $\cN = (0,2)$ theories, the ellipitic genus, have been pointed out in \cite{Bobev:2015kza}. It would be interesting to see if there are similar ambiguities for the topologically twisted index of $\cN=1$ gauge theories.}

Notice also that our results \eqref{Z large N conjecture0} and \eqref{centralchargea0} are compatible with  a very simple relation between the field theoretical quantities  $\Tr R^3 ( \Delta_I) $ and $c_{r} \left( \Delta_I , \fn_I \right)$ that is worthwhile to state separately,
\be
 \label{field theory}
 c_{r} \left( \Delta_I , \fn_I \right) = - 3 \Tr R^3 ( \Delta_I) - \pi \sum_{I} \left[ \left( \fn_I - \frac{\Delta_I}{\pi} \right) \frac{\partial \Tr R^3 ( \Delta_I)}{\partial \Delta_I} \right] \, .
\ee

The rest of the paper is organized as follows. In Section \ref{The topologically twisted index}
we review the basic properties of the topologically twisted index in four dimensions.
In Section \ref{SYM} we analyze the high-temperature limit of the index for $\cN=4$ super Yang-Mills
while in Section \ref{Klebanov-Witten} we discuss the example of the conifold.
Then in Section \ref{high-temp limit of the index} we derive the formulae \eqref{tHoof:anomaly:0}, \eqref{index theorem:2d central charge0},
\eqref{Z large N conjecture0} and \eqref{centralchargea0}.
The body of the paper ends with Section \ref{discussion}, which contains  possible future problems to explore.
In Appendix \ref{Elliptic functions} we derive the asymptotics of the elliptic functions relevant for our computations.
Appendix \ref{Anomaly cancellation} is devoted to the study of anomaly cancellation conditions for theories on $S^2 \times T^2$.

\section{The topologically twisted index}
\label{The topologically twisted index}

The topologically twisted index of an $\cN = 1$ gauge theory with vector and chiral multiplets
and a non-anomalous $\U(1)_R$ symmetry in four dimensions is defined
as the path-integral of the theory on $S^2 \times T^2$ with a partial topological $A$-twist along $S^2$ \cite{Benini:2015noa}.
It is a function of $q = e^{2 \pi i \tau}$, where $\tau$ is the modular parameter of $T^2$, fugacities $y$
for the global symmetries and flavor magnetic fluxes $\fn$ on $S^2$ parameterizing the twist.
The index can be reduced to a matrix integral over  zero-mode gauge variables  by exploiting the localization technique.
The zero-mode gauge variables  $x = e^{i u}$  parameterize the Wilson lines on the two directions of the torus 
\be
u = 2 \pi \oint_{\textmd{A-cycle}} A - 2 \pi \tau \oint_{\textmd{B-cycle}} A \, , 
\ee
and are defined modulo 
\be u_i \sim u_i + 2 \pi n + 2 \pi m \tau\, ,\qquad\qquad  n\, ,m \in \mathbb{Z} \, . \ee
Explicitly, for a theory with gauge group $G$ and a set of chiral multiplets transforming in representations $\fR_I$ of $G$,
the topologically twisted index is given by a contour integral of a meromorphic form\footnote{Supersymmetric localization
picks a particular contour of integration and the final result can be expressed in terms of the Jeffrey-Kirwan residue \cite{Benini:2015noa}.}
\bea
 \label{path integral index}
 Z (\fn, y) = \frac1{|\cW|} \; \sum_{\fm \,\in\, \Gamma_\fh} \; \oint_\cC \;  & \prod_{\text{Cartan}} \left (\frac{dx}{2\pi i x}  \eta(q)^{2} \right )
 (-1)^{\sum_{\alpha > 0}\alpha (\fm)} \prod_{\alpha \in G} \left[ \frac{\theta_1(x^\alpha ; q)}{i \eta(q)} \right] \\
 \times & \prod_I \prod_{\rho_I \in \fR_I} \bigg[ \frac{i \eta(q)}{\theta_1(x^{\rho_I} y_I ; q)} \bigg]^{\rho_I(\fm)- \fn_I + 1} \, ,
\eea
where $\alpha$ are the roots of $G$ and $|\cW|$ denotes the order of the Weyl group.
Given a weight $\rho_I$ of the representation $\fR_I$, we use the notation $x^{\rho_I} = e^{i \rho_I(u)}$.
In this formula, $\theta_1(x; q)$ is a Jacobi theta function and $\eta(q)$ is the Dedekind eta function (see Appendix \ref{Elliptic functions}).
The result is summed over a lattice of gauge magnetic fluxes $\fm$ on $S^2$ living in the co-root lattice $\Gamma_{\fh}$ of the gauge group $G$
(up to gauge transformations).
The integrand in \eqref{path integral index} is a well-defined meromorphic function on the torus provided that the gauge and the gauge-flavor anomalies vanish (see Appendix \ref{Anomaly cancellation}). 

The topologically twisted index \eqref{path integral index} depends on a choice of fugacities $y_I$ for the flavor group
and a choice of integer magnetic fluxes $\fn_I$ for the R-symmetry of the theory.
It is useful to introduce complex chemical potentials $y_I=e^{i \Delta_I}$.
In an $\cN = 1$ theory, the choice of the R-symmetry is not unique,
and can be mixed with the $\U(1)$ flavor symmetries
\be
\label{mixing}
\fn_I = r_I + \fp_I \, ,
\ee
where $r_I$ is a reference R-symmetry and $\fp_I$ are magnetic fluxes under the flavor symmetries of the theory.
The invariance of each monomial term $W$ in the superpotential under the symmetries of the theory imposes the following constraints
\be\label{supconstraints}
 \prod_{I \in W} y_I = 1 \, , \qquad \qquad \sum_{I \in W} \fn_I = 2 \, ,
\ee
where the latter comes from supersymmetry, and, as a consequence,
\be\label{supconstraints2}
 \sum_{I \in W} \Delta_I \in 2 \pi \mathbb{Z}  \, .
\ee
Here, the product and the sum are restricted to the fields entering in the monomial $W$.

\section{$\cN = 4$ super Yang-Mills}
\label{SYM}

We first consider the twisted compactification of four-dimensional $\cN = 4$ super Yang-Mills (SYM) with gauge group $\SU(N)$ on $S^2$.
At low energies, it results in a family of 2d theories with $\cN = (0,2)$ supersymmetry depending on the twisting parameters $\fn$  \cite{Benini:2012cz,Benini:2013cda}.
The theory describes the dynamics of $N$ D3-branes wrapped on $S^2$ and can be pictured as the quiver gauge theory given in \eqref{SYM:quiver}.
\bea
\label{SYM:quiver}
\begin{tikzpicture}[font=\footnotesize, scale=0.9]
\begin{scope}[auto,%
  every node/.style={draw, minimum size=0.5cm}, node distance=2cm];
\node[circle]  (UN)  at (0.3,1.7) {$N$};
\end{scope}
\draw[decoration={markings, mark=at position 0.45 with {\arrow[scale=2.5]{>}}, mark=at position 0.5 with {\arrow[scale=2.5]{>}}, mark=at position 0.55 with {\arrow[scale=2.5]{>}}}, postaction={decorate}, shorten >=0.7pt] (-0,2) arc (30:341:0.75cm);
\node at (-2.2,1.7) {$\phi_{1,2,3}$};
\end{tikzpicture}
\eea
The superpotential 
\be
 \label{SYM:superpotential}
 W = \Tr \left( \phi_3 \left[ \phi_1, \phi_2 \right] \right) 
\ee
imposes the following constraints on the chemical potentials $\Delta_a$ and the flavor magnetic fluxes $\fn_a$ associated with the fields $\phi_a$,
\be
 \label{SYM:constraints}
 \sum_{a = 1}^{3} \Delta_a \in 2 \pi \mathbb{Z} \, , \qquad \qquad \sum_{a = 1}^{3} \fn_a = 2 \, .
 \ee 
The topologically twisted index for the $\SU(N)$ SYM theory is given by
\begin{equation}
 \label{SYM path integral_constraint}
 Z = \frac{\cA}{N!} \;
 \sum_{\substack{\fm \,\in\, \mathbb{Z}^N \, , \\ \sum_i \fm_i = 0}} \; \int_\cC \;
 \prod_{i=1}^{N - 1} \frac{d x_i}{2 \pi i x_i} \prod_{j \neq i}^{N} \frac{\theta_1\left( \frac{x_i}{x_j} ; q\right)}{i \eta(q)}
 \prod_{a=1}^{3} \left[ \frac{i \eta(q)}{\theta_1\left( \frac{x_i}{x_j} y_a ; q\right)} \right]^{\fm_i - \fm_j - \fn_a + 1} \, ,
\end{equation}
where we defined the quantity
\be
\label{SYM:A}
\cA = \eta(q)^{2 (N-1)} \prod_{a=1}^{3} \left[ \frac{i \eta(q)}{\theta_1\left(y_a ; q\right)} \right]^{(N-1) (1 - \fn_a)} \, .
\ee
Here, we already imposed the $\SU(N)$ constraint $\prod_{i = 1}^{N} x_i = 1$.
Instead of performing a constrained sum over gauge magnetic fluxes we introduce 
the Lagrange multiplier $w$ and consider an unconstrained sum.
Thus, the index reads
\begin{equation}
 \label{SYM path integral}
 Z = \frac{\cA}{N!} \;
 \sum_{\fm \,\in\, \mathbb{Z}^N} \; \int_\cB \frac{d w}{2 \pi i w} w^{\sum_{i = 1}^{N} \fm_i} \; \int_\cC \;
 \prod_{i=1}^{N - 1} \frac{d x_i}{2 \pi i x_i} \prod_{j \neq i}^{N} \frac{\theta_1\left( \frac{x_i}{x_j} ; q\right)}{i \eta(q)}
 \prod_{a=1}^{3} \left[ \frac{i \eta(q)}{\theta_1\left( \frac{x_i}{x_j} y_a ; q\right)} \right]^{\fm_i - \fm_j - \fn_a + 1} \, .
\end{equation}
In order to evaluate (\ref{SYM path integral}), we employ the same strategy as in \cite{Benini:2015noa,Benini:2015eyy}.
The Jeffrey-Kirwan residue picks a middle-dimensional contour in $\left(\mathbb{C}^{*}\right)^N$.
We can then take a large positive integer $M$ and resum the contributions $\fm \leq M-1$.
Performing the summations we get
\bea
Z = \frac{\cA}{N!} \; \int_\cB \frac{d w}{2 \pi i w} \;
\int_\cC \; \prod_{i=1}^{N - 1} \frac{dx_i}{2 \pi i x_i} \; \prod_{i=1}^{N} \frac{\left( e^{i B_i} \right)^M}{e^{i B_i} - 1}
\prod_{j \neq i}^{N} \frac{\theta_1\left( \frac{x_i}{x_j} ; q\right)}{i \eta(q)}
\prod_{a=1}^{3} \left[ \frac{i \eta(q)}{\theta_1\left( \frac{x_i}{x_j} y_a ; q\right)} \right]^{ 1- \fn_a} \, ,
\label{SYM:index:2}
\eea
where we defined
\be
\label{iB}
e^{i B_i} = w \prod_{j = 1}^{N} \prod_{a=1}^{3} \frac{ \theta_1 \left( \frac{x_j}{x_i} y_a ; q \right)}
{\theta_1 \left( \frac{x_i}{x_j} y_a ; q \right)} \, .
\ee
In picking the residues, we need to insert a Jacobian in the partition function and evaluate everything else at the poles,
which are located at the solutions to the {\it Bethe ansatz equations} (BAEs),
\be
\label{BAE}
 e^{i B_i} = 1 \, ,
\ee
such that the off-diagonal vector multiplet contribution does not vanish. We consider \eqref{BAE} as a system of 
$N$ independent equations with respect to $N$ independent variables $\{x_1,\dots,x_{N-1},w\}$. 
In the final expression, the dependence on the cut-off $M$ disappears and we find
\be
Z = \cA \sum_{I \in\mathrm{BAEs}}\frac{1}{\mathrm{det}\mathds{B}} \prod_{j \neq i}^{N}\frac{\theta_1\left(\frac{x_i}{x_j};q \right)}{i\eta(q)}
\prod_{a = 1}^3\left[ \frac{i\eta(q)}{\theta_1\left(\frac{x_i}{x_j}y_a;q \right)} \right]^{1-\fn_a} \, ,
\label{index:SYM:bethe}
\ee
where the summation is over all solutions $I$ to the BAEs \eqref{BAE}.
The matrix $\mathds{B}$ appearing in the Jacobian has the following form
\be
\mathds{B} = \frac{\partial \left( e^{i B_1}, \ldots, e^{i B_N} \right)}
{\partial \left( \log x_1, \ldots, \log x_{N-1}, \log w \right)} \, .
\ee

\subsection{Bethe potential at high temperature}
\label{Bethe potential_SYM}

We also introduce the ``Bethe potential'', a function 
that has critical points at the solutions to the BAEs \eqref{BAE}.
Here, we will not give the general expression for the Bethe potential as it is quite involved.
Instead, we will try to make our problem easier by looking at the \emph{high-temperature} limit, \ie\;$q \to 1\; (\tau \to i 0)$.

Let us start by considering the BAEs \eqref{BAE} at high temperature.
Taking the logarithm of the BAEs \eqref{BAE}, we obtain
\bea
\label{SYM:BAE:logarithm}
0  = - 2 \pi i n_i + \log w - \sum_{j = 1}^{N} \sum_{a=1}^{3} \left\{ \log\left[ \theta\left(\frac{x_i}{x_j}y_a;q \right)\right]
- \log\left[ \theta\left(\frac{x_j}{x_i}y_a;q \right)\right] \right\} \, ,
\eea
where $n_i$ is an integer that parameterizes the angular ambiguity.
It is convenient to use the variables $u_i$, $\Delta_a$, $v$, defined modulo $2 \pi$:
\be
x_i = e^{i u_i}\, ,\qquad \qquad y_a = e^{i \Delta_a} \, ,\qquad \qquad w = e^{i v} \, .
\ee
Then, using the asymptotic formul\ae\;\eqref{dedekind:hight:S} and \eqref{theta:hight:S}
we obtain the high-temperature limit of the BAEs \eqref{SYM:BAE:logarithm}, up to exponentially suppressed corrections,
\bea
0= - 2 \pi i n_i + i v + \frac{1}{\beta}\sum_{j = 1}^{N} \sum_{a=1}^3 \left[ F' \left(u_i - u_j + \Delta_a\right) - F' \left(u_j - u_i + \Delta_a\right) \right] \, ,
\label{BAE:SYM:hight}
\eea
where $i / (2 \pi \tau) = 1/\beta$ is the formal ``temperature'' variable.
Here, we have introduced the polynomial functions
\bea
F(u) = \frac{u^3}{6} - \frac{1}{2}\pi u^2 \mathrm{sign}[\re(u)] + \frac{\pi^2}{3} u \, , \quad \quad 
F'(u) = \frac{u^2}{2} - \pi u \sign[\re(u)] + \frac{\pi^2}{3} \, .
\label{F:function}
\eea

The high-temperature limit of the Bethe potential can be found directly
by integrating the BAEs \eqref{BAE:SYM:hight} with respect to $u_i$ and summing over $i$.
It reads
\bea
\cV (\{u_i\}) & = \sum_{i = 1}^{N} \left( 2 \pi n_i - v \right) u_i
+ \frac{i (N - 1)}{\beta} \sum_{a = 1}^{3} F \left( \Delta_a \right)
\\ & + \frac{i}{2 \beta}\sum_{i \neq j}^{N} \sum_{a = 1}^{3}
 \left[ F \left( u_i - u_j + \Delta_a \right) + F \left( u_j - u_i + \Delta_a \right) \right]
 \, .
\label{Bethe:potential:SYM:hight}
\eea
It is easy to check that the BAEs \eqref{BAE:SYM:hight} can be obtained as critical points of the above Bethe potential. We introduced a $\Delta_a$-dependent 
integration constant in order to have precisely one contribution  $F \left( u_i - u_j + \Delta_a \right)$ for each component of the adjoint multiplet.

It is natural to restrict the $\Delta_a$ to the fundamental domain. In the high-temperature limit,  we can assume that $\Delta_a$ are real and $0 < \Delta_a < 2 \pi$.
Moreover,  since \eqref{SYM:constraints} must hold, $\sum_{a=1}^3 \Delta_a$ can only be $0,2\pi,4\pi$ or  $6\pi$.
We have checked that $\sum_{a=1}^3 \Delta_a = 0, 6\pi$ lead to a singular index, and
those for $2\pi$ and $4\pi$ are related by a discrete symmetry of the index  \ie\;$y_a \to 1 / y_a \left( \Delta_a \to 2 \pi - \Delta_a \right)$.
Thus, without loss of generality, we will assume $\sum_{a=1}^{3} \Delta_a = 2 \pi$ in the following. 
 
\paragraph*{The solution for $\boldsymbol{\sum_a \Delta_a = 2 \pi}$.} 
We seek for solutions to the BAEs \eqref{BAE:SYM:hight} assuming that
\be
 0 < \re \left( u_j - u_i \right) + \Delta_{a} < 2 \pi \, , \qquad \forall \quad i, j, a \, .
\ee
Thus, the high-temperature limit of the BAEs \eqref{BAE:SYM:hight} takes the simple form
\bea
\frac{2}{\beta}\sum_{a = 1}^{3} \left( \Delta_a-\pi \right)\sum_{k = 1}^{N} \left( u_j - u_k \right) = i \left( 2\pi n_j-v \right) \, ,
\quad \mbox{for} \quad j=1,2, \ldots, N \, .
\eea
Imposing the constraints $\sum_{a = 1}^{3} \Delta_a = 2 \pi$ for the chemical potentials as well as $\SU(N)$ constraint $\sum_{i=1}^{N} u_i=0$ 
we obtain the following set of equations
\bea
 \label{BAE:SYM:hight:simp}
 \frac{i N}{\beta} u_j & = n_j-\frac{v}{2\pi}\, ,\quad \mbox{for}\quad  j=1,\dots,N-1 \, , \\ 
 -\frac{i N}{\beta}\sum_{j=1}^{N-1} u_j & = n_N-\frac{v}{2\pi} \, .
 \eea
Summing up all equations we obtain the solution for $v$, which is given by 
\be
v=\frac{2\pi}{N}\sum_{i=1}^N n_i\, .
\label{SYM:solution:v}
\ee
The solution for eigenvalues $u_i$ reads 
\be
u_i=-\frac{i \beta}{N}\left( n_i-\frac{1}{N}\sum_{i=1}^N n_i \right)\,.
\label{SYM:solution:u}
\ee
Notice that, the tracelessness condition is automatically satisfied in this case.

To proceed further, we need to provide an estimate on the value of the constants $n_i$.
Whenever any two integers are equal $n_i = n_j$, we find that the off-diagonal vector multiplet contribution to the index,
which is an elliptic generalization of the Vandermonde determinant, vanishes.
Moreover, the high-temperature expansion \eqref{theta:hight:S} breaks down as subleading terms start blowing up.
Hence, we should make another ansatz for the phases $n_i$ such that
\be
n_i-n_j\neq 0 ~ ~ \mbox{mod} ~ ~ N \, .
\label{condition:n}
\ee
To understand how much freedom we have, let us first note that 
eigenvalues $u_i$ are variables defined on the torus $T^2$ and thus they should be periodic in $\beta$.
Due to \eqref{SYM:solution:u}, this means that integers $n_i$ are defined modulo $N$ and hence, without loss 
of generality, we can consider only integers lying in the domain $\left[ 1,N \right]$ with the condition 
\eqref{condition:n} modified to $n_i\neq n_j \, ,\forall~ i,j$.
This leaves us with the only choice $n_i=i$ and its permutations.

Substituting \eqref{SYM:solution:u} and \eqref{SYM:solution:v}
into the Bethe potential \eqref{Bethe:potential:SYM:hight}, we obtain
\be
 \label{SYM:onshel:bethe potential}
 \cV \left( \Delta_a \right) \big|_{\text{BAEs}} = \frac{i \left( N^2 - 1 \right)}{\beta} \sum_{a = 1}^{3} F \left( \Delta_a \right)
 = \frac{i \left( N^2 - 1 \right)}{2 \beta} \Delta_1 \Delta_2 \Delta_3 \, ,
\ee
up to terms $\cO(\beta)$.

There is an interesting relation between the ``on-shell'' Bethe potential \eqref{SYM:onshel:bethe potential} and the central charge of
the UV four-dimensional theory. Note that, given the constraint $\sum_{a=1}^{3} \Delta_{a} = 2 \pi$, the quantities $\Delta_a$ can be used to parameterize the most general R-symmetry of the theory
\be R(\Delta_a) = \sum_{a=1}^3  \Delta_a  \frac{R_a}{2\pi} \, ,\ee
where $R_a$ gives charge 2 to $\phi_a$ and zero to $\phi_b$ with $b\neq a$. Observe also that the cubic R-symmetry 't Hooft anomaly is given by
\bea
 \Tr R^3 \left( \Delta_a \right) & = \left(N^2 - 1\right) \left[ 1 + \sum _{a = 1}^{3} \left( \frac{\Delta_a}{\pi} - 1 \right)^3 \right] \\
 & = \frac{3 \left( N^2 - 1 \right)}{\pi^3} \Delta_1 \Delta_2 \Delta_3 \, ,
\eea
where the trace is taken over the fermions of the theory.
Therefore, the ``on-shell'' value of the Bethe potential \eqref{SYM:onshel:bethe potential} can be rewritten as
\be
\label{on:shell:Bethe:pot:hight:SYM}
\cV \left(\Delta_a \right) \big|_{\text{BAEs}} = \frac{i \pi^3}{6 \beta} \Tr R^3 \left(\Delta_a \right) = \frac{16 i \pi^3}{27 \beta} a \left(\Delta_a \right) \, ,
\ee
where in the second equality we used the relation \eqref{generalac}. Note that the linear R-symmetry 't Hooft anomaly is zero for $\cN = 4$ SYM. 

\subsection{The topologically twisted index at high temperature}
\label{The index at high temperature_SYM}

We are interested in the high-temperature limit of the logarithm of the partition function \eqref{index:SYM:bethe}.
We shall use the asymptotic expansions \eqref{dedekind:hight:S} and \eqref{theta:hight:S} in order to calculate the vector and hypermultiplet contributions to the twisted index in the $\beta \to 0$ limit.

The contribution of the off-diagonal vector multiplets can be computed as
\begin{align}
\log \prod_{i \neq j}^N \left[ \frac{\theta_1\left( \frac{x_i}{x_j} ; q\right)}{i \eta(q)}\right] =
- \frac{1}{\beta} \sum_{i \neq j}^{N} F' \left( u_i - u_j \right) - \frac{ i N (N - 1) \pi}{2} \, ,
\end{align}
in the asymptotic limit $q \to 1\;(\beta \to 0)$. The contribution of the matter fields is instead
\bea
 \log \prod_{i \neq j}^N \prod_{a=1}^3 \left[ \frac{i\eta(q)}{\theta_1\left(\frac{x_i}{x_j}y_a;q \right)} \right]^{1-\fn_a} & =
 - \frac{1}{\beta} \sum_{i \neq j}^{N} \sum_{a = 1}^{3} \left[ \left(\fn_a - 1\right) F' \left( u_i - u_j + \Delta_a \right) \right] \\
 & + \frac{ i N (N-1) \pi}{2} \sum_{a = 1}^{3} \left( 1 -  \fn_a \right) \, , \qquad \textmd{as } \; \beta \to 0 \, .
\eea
The prefactor $\cA$ in the partition function \eqref{SYM:A} at high temperature contributes
\bea
 \log \left\{ \eta(q)^{2 (N-1)} \prod_{a=1}^{3} \left[ \frac{i \eta(q)}{\theta_1\left(y_a ; q\right)} \right]^{(N-1) (1 - \fn_a)} \right\} & =
 - \frac{N-1}{\beta} \left[ \frac{\pi^2}{3} + \sum_{a = 1}^{3} \left( \fn_a - 1 \right) F' (\Delta_a) \right] \\
 & - (N-1) \left[ \log \left( \frac{\beta}{2 \pi}\right) - \frac{i \pi}{2} \sum_{a = 1}^{3} \left( 1 -  \fn_a \right) \right] \, .
\eea
The last term to work out is  $ - \log \det \mathds{B}$.
The matrix $\mathds{B}$, imposing $e^{i B_i} = 1$, reads 
\bea
\mathds{B} = \frac{\partial \left( B_1, \ldots, B_N \right)}{\partial \left( u_1, \ldots, u_{N-1}, v \right)}
\, , \qquad \textmd{as } \; \beta \to 0 \, ,
\label{Jacobian:SYM}
\eea
and has the following entries 
\bea
\frac{\partial B_k}{\partial u_j} & = \frac{2\pi i}{\beta}N\delta_{kj} \, ,\quad \mbox{for} \quad k,j=1,2,\dots,N-1\, ,\\
\frac{\partial B_N}{\partial u_k} & = -\frac{2\pi i}{\beta}N \, ,\quad \frac{\partial B_k}{\partial v}=1 \, ,
\quad \mbox{for} \quad k = 1, 2, \dots,N-1\, ,\\
\frac{\partial B_N}{\partial v} &= 1\, .
\eea
Here, we have already imposed the constraint $\sum_{a = 1}^{3} \Delta_a = 2 \pi$.
Therefore, we obtain
\be
- \log \det \mathds{B} = (N-1) \left[ \log \left( \frac{\beta}{2 \pi} \right) - \frac{i \pi}{2}\right] - N\log N  \, .
\ee
Putting everything together we can write the high-temperature limit of the twisted index, at finite $N$,
\bea
\label{SYM:index:final:hight}
 \log Z & = - \frac{1}{\beta} \sum_{i \neq j}^{N} \left[ F' \left( u_i - u_j \right) + \sum_{a=1}^3 \left( \fn_a - 1\right) F' \left(u_i - u_j + \Delta_a \right) \right] \\
 & - \frac{N - 1}{\beta} \left[ \frac{\pi^2}{3} + \sum_{a=1}^3 \left( \fn_a - 1 \right) F' \left( \Delta_a \right) \right] - N \log N \, ,
\eea
up to exponentially suppressed corrections. We may then evaluate the index by
substituting the pole configurations \eqref{SYM:solution:u} back into the functional \eqref{SYM:index:final:hight} to get,
\bea
 \label{SYM:logZ:bethe}
  \log Z & = - \frac{N^2 - 1}{\beta} \left[ \frac{\pi^2}{3} + \sum_{a=1}^3 \left( \fn_a - 1 \right) F' \left( \Delta_a \right) \right] - N \log N  \\
& = - \frac{N^2 - 1}{2 \beta} \sum_{\substack{ a < b \\ (\neq c) }} \Delta_a \Delta_b \fn_c - N \log N \, ,
\eea
which, to leading order in $1 / \beta$, can be rewritten in a more intriguing form:
\be
\label{SYM:index:final:hight:bethe}
 \log Z = i \sum_{a=1}^{3} \fn_a \frac{\partial \cV(\Delta_a) \big|_{\text{BAEs}} }{\partial \Delta_a} \, .
\ee

We can relate the index to the trial left-moving central charge of the two-dimensional $\cN = (0, 2)$ theory on $T^2$. The latter reads \cite{Benini:2012cz,Benini:2013cda}
\be
c_l = c_r - k \, ,
\ee
where $k$ is the gravitational anomaly
\be
k = - \Tr \gamma_3 = - \left( N^2 - 1 \right) \left[ 1 + \sum_{a = 1}^{3} \left( \fn_a - 1 \right) \right] = 0 \, ,
\ee
and $c_r$ is the trial right-moving central charge
\bea
 \label{c2d:anomaly}
 c_{r} \left( \Delta_a , \fn_a \right) = - 3 \Tr  \gamma_3 R^2 \left( \Delta_a \right)
 &= - 3 \left( N^2 - 1 \right) \left[ 1 + \sum_{a=1}^{3} \left( \fn_a - 1 \right) \left( \frac{\Delta_a}{\pi} - 1 \right)^2 \right] \,  \\
 &= - \frac{3 \left( N^2 - 1 \right)}{\pi^2}  \sum_{\substack{ a < b \\ (\neq c) }} \Delta_a \Delta_b \fn_c \, .
\eea
Here, the trace is taken over the fermions and $\gamma_3$ is the chirality operator in 2d. In the twisted compactification, each of the chiral fields $\phi_a$ give rise to 2d fermions.
The difference between the number of fermions of opposite chiralities is $\fn_a-1$, thus explaining the above formulae.
We used $\Delta_a/\pi$ to parameterize the trial R-symmetry.   
We find that the index is given by
\be
\label{SYM:index:final:hight:bethe2}
\label{SYM:2d:4d:relation}
 \log Z = \frac{\pi^2}{6 \beta} c_{r} \left(\Delta_a , \fn_a \right)
 = - \frac{16 \pi^3}{27\beta} \sum_{a=1}^{3} \fn_a \frac{\partial a(\Delta_a) }{\partial \Delta_a} \, .
\ee
As shown in \cite{Benini:2012cz,Benini:2013cda}, the exact central charge of the 2d CFT is obtained by extremizing $c_{r} \left( \Delta_I , \fn_I \right)$ with respect to the $\Delta_I$.
Given the above relation \eqref{SYM:2d:4d:relation}, we see that this is equivalent to extremizing the $\log Z$ at high temperature.

\section{The Klebanov-Witten theory}
\label{Klebanov-Witten}

In this section we study the Klebanov-Witten theory \cite{Klebanov:1998hh} describing the low energy dynamics of $N$ D3-branes at the conifold singularity.
This is the Calabi-Yau cone over the homogeneous Sasaki-Einstein five-manifold $T^{1,1}$ which can be described as the coset $\SU(2) \times \SU(2) / \U(1)$.
This theory has $\cN = 1$ supersymmetry and has $\SU(N) \times \SU(N)$ gauge group with bi-fundamental chiral multiplets $A_i$ and $B_j$, $i , j = 1, 2$, transforming in the $\left( {\bf N}, \overline{\bf N} \right)$ and $\left( \overline{\bf N}, {\bf N} \right)$ representations of the two gauge groups. This can be pictured as
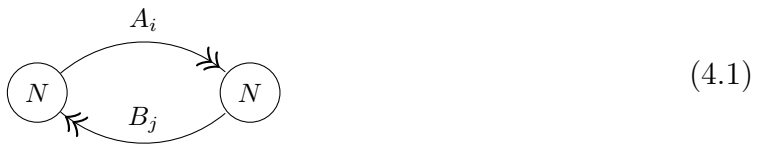
\bea
\label{KW:quiver}
\begin{tikzpicture}[baseline, font=\footnotesize, scale=0.8]
\begin{scope}[auto,%
  every node/.style={draw, minimum size=0.5cm}, node distance=2cm];
\node[circle] (USp2k) at (-0.1, 0) {$N$};
\node[circle, right=of USp2k] (BN)  {$N$};
\end{scope}
\draw[decoration={markings, mark=at position 0.9 with {\arrow[scale=2.5]{>}}, mark=at position 0.95 with {\arrow[scale=2.5]{>}}}, postaction={decorate}, shorten >=0.7pt]  (USp2k) to[bend left=40] node[midway,above] {$A_{i}$} node[midway,above] {} (BN) ; 
\draw[decoration={markings, mark=at position 0.1 with {\arrow[scale=2.5]{<}}, mark=at position 0.15 with {\arrow[scale=2.5]{<}}}, postaction={decorate}, shorten >=0.7pt]  (USp2k) to[bend right=40] node[midway,above] {$B_{j}$}node[midway,above] {}  (BN) ;  
\end{tikzpicture}
\eea
It has a quartic superpotential,
\be
\label{KW:superpotential}
W = \Tr \left( A_1 B_1 A_2 B_2 - A_1 B_2 A_2 B_1 \right) \, .
\ee
We assign chemical potentials $\Delta_{1,2} \in (0, 2 \pi)$ to $A_i$ and $\Delta_{3,4} \in (0, 2 \pi)$ to $B_i$.
Invariance of the superpotential under the global symmetries requires
\begin{align}
\label{KW:constraints}
\sum_{I = 1}^{4} \Delta_I \in 2 \pi \mathbb{Z} \, , \qquad \qquad \sum_{I = 1}^{4} \fn_I = 2 \, .
\end{align}
For the Klebanov-Witten theory, the topologically twisted index can be written as
\begin{multline}
\label{initial Z full}
Z = \frac1{(N!)^2} \sum_{\fm, \wt\fm \in \bZ^N} \int_\cB \; \frac{d w}{2 \pi i w}\frac{d \wt w}{2\pi i \wt w } w^{\sum_{i = 1}^N  \fm_i}\, 
\wt w^{ \sum_{i=1}^N\wt\fm_i  }  \times \\
\times \int_\cC \; \prod_{i=1}^{N-1}  \frac{dx_i}{2\pi i x_i} \,  \frac{d\tilde x_i}{2\pi i \tilde x_i} \eta(q)^{4 (N-1)}
\prod_{i\neq j}^N \left[ \frac{\theta_1\left( \frac{x_i}{x_j} ; q\right)}{i \eta(q)} \frac{\theta_1\left( \frac{\tilde x_i}{\tilde x_j} ; q\right)}{i \eta(q)}\right] \times \\
\times \prod_{i,j=1}^N \prod_{a=1,2}
\left[ \frac{i \eta(q)}{\theta_1\left( \frac{x_i}{\tilde x_j} y_a ; q\right)} \right]^{\fm_i - \wt\fm_j - \fn_a +1}
\prod _{b=3,4} \left[ \frac{i \eta(q)}{\theta_1\left( \frac{\tilde x_j}{x_i} y_b ; q\right)} \right]^{\wt\fm_j - \fm_i - \fn_b +1} \, .
\end{multline}
Here, we assumed that eigenvalues $x_i$ and $\tx_i$ satisfy the $\SU(N)$ constraint $\prod_{i=1}^N x_i=\prod_{i=1}^N \tx_i=1$. 
Hence, the integration is performed over $(N-1)$ variables instead of $N$. In order to impose 
the $\SU(N)$ constraints for the magnetic fluxes, \ie\;
\bea
\label{KW:tracelessness condition}
\sum_{i = 1}^{N} \fm_i = \sum_{i = 1}^{N} \wt \fm_i = 0 \, ,
\eea
we have introduced two Lagrange multipliers $w = e^{i v}$ and $\wt w=e^{i\tv}$.
Now, we can resum over gauge magnetic fluxes $\fm_i \leq M - 1$ and $\wt \fm_j \geq 1 - M$ 
for some large positive integer cut-off $M$. We obtain 
\bea
Z = \frac1{(N!)^2}  \int_\cB & \; \frac{d w}{2 \pi i w} \, \frac{d \wt w}{2\pi i \wt w } \int_\cC \; \prod_{i=1}^{N-1} \frac{dx_i}{2\pi i x_i} \, \frac{d\tilde x_i}{2\pi i \tilde x_i}
\prod_{i\neq j}^N \left[ \frac{\theta_1\left( \frac{x_i}{x_j} ; q\right)}{i \eta(q)} \frac{\theta_1\left( \frac{\tilde x_i}{\tilde x_j} ; q\right)}{i \eta(q)}\right] \times \\
& \times \cP \, \prod_{i=1}^N \frac{ (e^{iB_i})^M}{ e^{iB_i} -1 } \prod_{j=1}^N \frac{ (e^{i\wt B_j})^M}{ e^{i\wt B_j} - 1} \, ,
\eea
where we defined the quantities
\be
\label{KW:A}
\cP = \eta(q)^{4 (N-1)} \prod_{i,j=1}^N \, \prod_{a=1,2} \left[ \frac{i \eta(q)}{\theta_1\left( \frac{x_i}{\tilde x_j} y_a ; q\right)} \right]^{1-\fn_a}
\prod_{b=3,4} \left[ \frac{i \eta(q)}{\theta_1\left( \frac{\tilde x_j}{x_i} y_b ; q\right)} \right]^{1-\fn_b} \, ,
\ee
and
\bea
\label{BA expressions}
e^{iB_i} = w\prod_{j=1}^N \frac{\prod _{b=3,4} \theta_1\left( \frac{\tilde x_j}{x_i} y_b ; q\right)}
{\prod _{a=1,2} \theta_1\left( \frac{x_i}{\tilde x_j} y_a ; q\right)}  \, ,\qquad
e^{i\wt B_j} = \wt w^{-1} \,\prod_{i=1}^N   \frac{\prod _{b=3,4} \theta_1\left( \frac{\tilde x_j}{x_i} y_b ; q\right)}{\prod _{a=1,2}
\theta_1\left( \frac{x_i}{\tilde x_j} y_a ; q\right)} \, .
\eea
Then, similarly to the case of $\mathcal{N}=4$ SYM, the following BAEs equations
\be
\label{KW:BAEs}
e^{i B_i} = 1 \, ,\qquad \qquad \qquad e^{i \wt B_j} = 1 \, .
\ee
determine the poles of the integrand.
In order to calculate the index we simply insert a Jacobian 
of the transformation from $\{\log x_i,\log \tx_i, \log w,\log \wt w\}$ to $\{e^{i B_i},e^{i \tilde{B}_i}\}$ variables and evaluate everything else  
at the solutions to BAEs. In the final expression, the dependence on the cut-off $M$ disappears.
We can then write the partition function as,
\be
\label{index:KW:bethe}
Z = \sum_{I \in\mathrm{BAEs}}\frac{1}{\mathrm{det}\mathds{B}}
\prod_{i\neq j}^N \left[ \frac{\theta_1\left( \frac{x_i}{x_j} ; q\right)}{i \eta(q)} \frac{\theta_1\left( \frac{\tilde x_i}{\tilde x_j} ; q\right)}{i \eta(q)}\right] \cP \, ,
\ee
where $\mathds{B}$ is a $2N \times 2N$ matrix
\be
\label{Jacobian general:KW}
\mathds{B} = \frac{\partial \big( e^{iB_1},\dots,e^{i B_N}, e^{i\wt B_1},\dots, e^{i\wt B_N} \big) }
{ \partial\left( \log x_1,\dots,\log x_{N-1},\log w, \log \tx_1,\dots, \log \tx_{N-1}, \log \wt w\right)} \, .
\ee

\subsection{Bethe potential at high temperature}
\label{The Bethe potential_KW}

Let us now look at the Bethe potential at high temperature, \ie\;$\beta \to 0$ limit.
Taking the logarithm of the BAEs \eqref{KW:BAEs} we obtain
\bea
\label{KW:BAE:logarithm}
0 & = - 2 \pi i n_i + \log w - \sum_{j = 1}^{N} \left\{ \sum_{a = 1,2} \log\left[ \theta_1\left( \frac{x_i}{\tilde x_j} y_a ; q\right)\right]
- \sum_{b = 3,4} \log\left[ \theta_1\left( \frac{\tilde x_j}{x_i} y_b ; q\right)\right]  \right\} \, , \\
0 & = - 2 \pi i \tilde n_j - \log \wt w - \sum_{i = 1}^{N} \left\{ \sum_{a = 1,2} \log\left[ \theta_1\left( \frac{x_i}{\tilde x_j} y_a ; q\right) \right]
- \sum_{b = 3,4} \log\left[ \theta_1\left( \frac{\tilde x_j}{x_i} y_b ; q\right) \right]  \right\} \, ,
\eea
where $n_i\, , \tilde n_j$ are integers that parameterize the angular ambiguities.
In order to compute the high-temperature limit of the above BAEs, we go to the variables $u_i \, , \tilde u_j \, , \Delta_I \, , v\, , \tv$,
defined modulo $2 \pi$,
and employ the asymptotic expansions \eqref{dedekind:hight:S} and \eqref{theta:hight:S}.
We find
\bea
\label{BAE:KW:hight}
0 & = -2\pi i n_i +iv+ \frac{1}{\beta} \sum_{j = 1}^{N} \left[ \sum_{a = 1,2} F' \left(u_i - \tilde u_j + \Delta_a\right) 
- \sum_{b = 3,4} F' \left(\tilde u_j - u_i + \Delta_b\right)  \right] 
\, , \\
0 & = -2\pi i \tilde n_j - i \tv +  \frac{1}{\beta} \sum_{i = 1}^{N} \left[ \sum_{a = 1,2} F' \left(u_i - \tilde u_j + \Delta_a\right) 
- \sum_{b = 3,4} F' \left(\tilde u_j - u_i + \Delta_b\right)\right] 
\, ,
\eea
where the polynomial function $F' (u)$ is defined in \eqref{F:function}.
The BAEs \eqref{BAE:KW:hight} can be obtained as critical points of the Bethe potential
\bea
\label{KW:bethe potential}
\begin{aligned}
 \cV \left(\{u_i, \tilde u_i\}\right) & = 2\pi \sum_{i = 1}^{N} \left( n_i u_i - \tilde n_i \tilde u_i \right) -
  \sum_{i=1}^{N} \left( v\, u_i+\tv\,\tu_i \right) \\
  & + \frac{i}{\beta} \sum_{i,j=1}^{N} \left[ \sum_{a=1,2} F(u_i - \tilde u_j + \Delta_a) + \sum_{b=3,4} F(\tilde u_j - u_i + \Delta_b) \right]\, .
\end{aligned}
\eea

We next turn to find solutions to the BAEs \eqref{BAE:KW:hight}. The constraints (\ref{KW:constraints}) imply that 
$\sum_{I=1}^4 \Delta_I$ can only be $0,2\pi,4\pi,6\pi$ or  $8\pi$.
For the conifold theory, it turns out that the solutions with $\sum_{I=1}^4 \Delta_I = 0, 8\pi$ lead to a singular index,
those for $2 \pi$ and $6 \pi$ are related by a discrete symmetry of the index, \ie\;$y_I \to 1 / y_I \left( \Delta_I \to 2 \pi - \Delta_I \right)$, and
there are no consistent solutions for $\sum_{I=1}^4 \Delta_I = 4\pi$.
Thus, without loss of generality, we  assume again $\sum_{I=1}^4 \Delta_I = 2 \pi$ in the following. 

\paragraph*{The solution for $\boldsymbol{\sum_I \Delta_I = 2 \pi}$.}
We assume that
\bea
0 < \re \left( \tilde u_j - u_i \right) + \Delta_{3,4} < 2 \pi \, , \qquad \qquad - 2 \pi < \re \left( \tilde u_j - u_i \right) - \Delta_{1,2} < 0 \, , \qquad \forall \quad i, j \, .
\eea
Hence, the BAEs \eqref{BAE:KW:hight} become
\bea
\label{BAE:KW:hight:simplified}
0 & = - 2 \pi i n_j + i\,v- \frac{1}{\beta} \sum_{k = 1}^{N} \left[ \Delta_1 \Delta_2 - \Delta_3 \Delta_4 - 2 \pi  \left(\tilde u_k - u_j \right) \right] \, , \\
0 & = - 2 \pi i \tilde n_k -i\,\tv -\frac{1}{\beta} \sum_{j = 1}^{N} \left[ \Delta_1 \Delta_2 - \Delta_3 \Delta_4 - 2 \pi  \left(\tilde u_k - u_j \right)  \right] \, .
\eea
Here, we have already imposed the constraint $\sum_{I=1}^{4} \Delta_I = 2 \pi$.
Imposing the $\SU(N)$ constraints for $u_i\,,\, \tu_i$ we can rewrite the BAEs in the following form 
\be
\label{BAE:KW:1}
\frac{iN}{\beta}\,u_j &=& n_j-\frac{v}{2\pi}+\frac{iN}{2\pi\beta}\left( \Delta_3\,\Delta_4-\Delta_1\,\Delta_2 \right)\, ,\quad \mbox{for}\quad j=1,\dots,N-1\, ,\\
\label{BAE:KW:2}
-\frac{iN}{\beta}\,\sum_{j=1}^{N-1} u_j &=& n_N-\frac{v}{2\pi}+\frac{iN}{2\pi\beta}\left( \Delta_3\,\Delta_4-\Delta_1\,\Delta_2 \right)\, , \\
\label{BAE:KW:3}
\frac{iN}{\beta}\,\tu_j &=& - \tilde n_j-\frac{\tv}{2\pi}-\frac{iN}{2\pi\beta}\left( \Delta_3\,\Delta_4-\Delta_1\,\Delta_2 \right)\, ,\quad \mbox{for}\quad j=1,\dots,N-1\, ,\\
\label{BAE:KW:4}
-\frac{iN}{\beta}\,\sum_{j=1}^{N-1} \tu_j &=& -\tilde n_N-\frac{\tv}{2\pi}-\frac{iN}{2\pi\beta}\left( \Delta_3\,\Delta_4-\Delta_1\,\Delta_2 \right)\, .
\ee

Equations \eqref{BAE:KW:1} and \eqref{BAE:KW:3} can be considered as equations defining $u_i$ and $\tu_i$. In order to find $v$ and $\tv$ we need
to sum $(N-1)$ equations \eqref{BAE:KW:1} with \eqref{BAE:KW:2} and equations  \eqref{BAE:KW:3} with \eqref{BAE:KW:4}. This leads to 
\bea
\label{KW:sol:final}
v & =\frac{iN}{\beta}\left( \Delta_3\,\Delta_4-\Delta_1\,\Delta_2 \right)+\frac{2\pi}{N}\sum_{j=1}^N n_j \, , \qquad
u_j =-\frac{i\beta}{N}\left( n_j-\frac{1}{N}\sum_{i=1}^N n_i \right) \, , \\
\tilde v & =-\frac{iN}{\beta}\left( \Delta_3\,\Delta_4-\Delta_1\,\Delta_2 \right)-\frac{2\pi}{N}\sum_{j=1}^N \tilde n_j \, ,\qquad
\tilde u_j =\frac{i\beta}{N}\left( \tilde n_j-\frac{1}{N}\sum_{i=1}^N \tilde n_i \right)\, .
\eea
According to our prescription, all solutions which lead to zeros of the off-diagonal vector multiplet should be avoided.
Therefore, the allowed parameter space for integers $n_i$ and $\tilde n_i$ is determined by
\be
\label{KW:integers:constraint}
 n_j - n_i \neq 0 \quad \text{ mod }~ N \, , \qquad \qquad \tilde n_j - \tilde n_i \neq 0 \quad \text{ mod }~ N \, .
\ee
Given the solution \eqref{KW:sol:final} to the BAEs, the integers $n_i$ and $\tilde n_i$ are defined modulo $N$ due to the $\beta$-periodicity of eigenvalues on $T^2$.
Thus we are left with $\{n_i, \tilde n_i\} \in \left[ 1,N \right]$.
The only possible choice is then given by $n_i = \tilde n_i = i$ and its permutations.

Finally, plugging the solution \eqref{KW:sol:final} to the BAEs back into \eqref{KW:bethe potential},
we obtain the ``on-shell" value of the Bethe potential
\be
\label{SU(N):KW:on-shell:bethe potential}
\cV(\Delta_I) \big|_{\text{BAEs}} = \frac{i N^2}{2\beta} \sum_{a < b < c} \Delta_a \Delta_b \Delta_c \, ,
\ee
up to terms $\cO(\beta)$.
The relation between the ``on-shell'' Bethe potential and the 4d conformal anomaly coefficients also holds for the conifold theory.
The R-symmetry 't Hooft anomalies can be expressed as
\bea\label{KWR3}
 \Tr R \left( \Delta_I \right) & = 2 \left(N^2 - 1\right) + N^2 \sum _{I = 1}^{4} \left( \frac{\Delta_I}{\pi} - 1 \right) = - 2 \, , \\
 \Tr R^3 \left( \Delta_I \right) & = 2 \left(N^2 - 1\right) + N^2 \sum _{I = 1}^{4} \left( \frac{\Delta_I}{\pi} - 1 \right)^3 \\
 & = \frac{3 N^2}{\pi^3} \sum_{a < b < c} \Delta_a \Delta_b \Delta_c - 2 \, ,
\eea
where we used $\Delta_I / \pi$ to parameterize the trial R-symmetry of the theory.
Hence, Eq.\,\eqref{SU(N):KW:on-shell:bethe potential} can be rewritten as
\be
\label{SU(N):KW:on-shell:bethe potential:TrR3}
\cV(\Delta_I) \big|_{\text{BAEs}} = \frac{i \pi^3}{6 \beta} \left[ \Tr R^3 \left( \Delta_I \right) - \Tr R \left( \Delta_I \right) \right]
= \frac{16 i \pi^3}{27 \beta}  \left[ 3 c \left( \Delta_I \right) - 2 a( \Delta_I) \right] \, .
\ee
Here, we employed Eq.\,\eqref{generalac} to write the second equality.

\subsection{The topologically twisted index at high temperature}
\label{The index at high temperature_KW}

The twisted index, at high temperature, can be computed by evaluating the contribution of the saddle point configurations to \eqref{index:KW:bethe}.
The procedure for computing the index is very similar to that presented in section \ref{The index at high temperature_SYM}.
The off-diagonal vector multiplet contributes
\begin{align}
\log \prod_{i\neq j}^N \left[ \frac{\theta_1\left( \frac{x_i}{x_j} ; q\right)}{i \eta(q)} \frac{\theta_1\left( \frac{\tilde x_i}{\tilde x_j} ; q\right)}{i \eta(q)}\right] =
- \frac{1}{\beta} \sum_{i \neq j}^{N} \left[ F' \left( u_i - u_j \right) + F' \left( \tu_i - \tu_j \right) \right] - i N (N - 1) \pi \, .
\end{align}
The quantity $\cP$, Eq.\,\eqref{KW:A}, contributes
\bea
 \log \cP & = - \frac{1}{\beta} \bigg\{ \frac{2\pi^2}{3} (N-1)
 + \sum_{i , j = 1}^{N} \sum_{ \substack{I = 1,2: + \\ I = 3, 4: -}} (\fn_I - 1) F' \left[ \pm \left( u_i - \tilde u_j \pm \Delta_I \right) \right] \bigg\} \\  
 & + \frac{i N^2 \pi}{2} \sum_{I = 1}^{4} \left( 1 - \fn_I \right) - 2 (N-1) \log \left( \frac{\beta}{2 \pi} \right) \, .
\eea
The Jacobian \eqref{Jacobian general:KW} has the following entries
\bea
\frac{\partial B_k}{\partial u_j} & = -\frac{\partial \wt B_k}{\partial \tu_j} =\frac{2\pi i}{\beta}N\delta_{kj} \, ,\quad \mbox{for} \quad k,j=1,2,\dots,N-1\, ,\\
\frac{\partial B_N}{\partial u_k} & = -\frac{\partial \wt B_N}{\partial \tu_k}=-\frac{2\pi i}{\beta}N \, ,\quad 
\frac{\partial B_k}{\partial v}=-\frac{\partial \wt B_k}{\partial \tv}=1 \, ,
\quad \mbox{for} \quad k = 1, 2, \dots,N-1\, ,\\
\frac{\partial B_N}{\partial v} &=- \frac{\partial \wt B_N}{\partial \tv}=1\, ,\quad
\frac{\partial B_k}{\partial \tu_j}  = \frac{\partial \wt B_k}{\partial u_j} = \frac{\partial B_k}{\partial \tv} = \frac{\partial \wt B_k}{\partial v} =0\, ,\quad
\mbox{for}\quad k,j=1,\dots,N\, .
\eea
Now, it is straightforward to find the determinant of the matrix $\mathds{B}$:
\be
- \log \det \mathds{B} =2 (N-1) \left[ \log \left( \frac{\beta}{2 \pi} \right) - \frac{i \pi}{2}\right] -2 N \log N+\pi i N  \, .
\ee

The high-temperature limit of the index, at finite $N$, may then be written as
\bea
 \label{KW:index:final:hight}
 \log Z & = - \frac{1}{\beta} \bigg\{ \sum_{i \neq j}^{N} \left[ F' \left( u_i - u_j \right) + F' \left( \tu_i -\tu_j \right) \right]
 + \frac{2 \pi^2}{3} (N-1)
 \\ & + \sum_{i,j=1}^{N} \sum_{ \substack{I = 1,2: + \\ I = 3, 4: -}} (\fn_I - 1) F' \left[ \pm \left( u_i - \tilde u_j \pm \Delta_I \right) \right] \bigg\}
 \\ & - 2 N \log N + \pi i ( N + 1 ) \, .
\eea
Plugging the solution \eqref{KW:sol:final} to the BAEs back into the index \eqref{KW:index:final:hight} we find
\bea
\label{KW:index:final:hight:bethe:2d central charge}
\log Z = - \frac{N^2}{2 \beta} \sum_{ \substack{a<b \\ (\neq c)}} \Delta_a \Delta_b \fn_c  + \frac{2 \pi^2}{3 \beta} - 2 N \log N + \pi i ( N + 1 ) \, .
\eea
As in the case of $\cN=4$ SYM we can also write, to leading order in $1 / \beta$,
\be
\label{conifold:index:final:hight:bethe2}
 \log Z = \frac{\pi^2}{6 \beta} c_{l} \left(\Delta_I , \fn_I\right)
 = - \frac{16 \pi^3}{27 \beta}  \sum_{I=1}^{4} \fn_I \frac{\partial a (\Delta_I)}{\partial \Delta_I} \, ,
\ee
where the second equality is written assuming that $N$ is large.
Here, $c_l$ is the left-moving central charge of the 2d $\cN = (0, 2)$ SCFT obtained by the twisted compactification on $S^2$.
This is related to the trial right-moving central charge $c_r$ by the gravitational anomaly, \ie\;$c_l = c_r - k$.
The central charge $c_r$ takes contribution from the 2d massless fermions, the gauginos
and the zero-modes of the chiral fields (the difference between the number of modes of opposite chirality being  $\fn_I-1$) \cite{Benini:2012cz,Benini:2013cda,Benini:2015bwz},
\bea
 \label{c2d:anomalyconifold}
 c_{r} \left( \Delta_I , \fn_I \right) = - 3 \Tr  \gamma_3 R^2 \left( \Delta_I \right)
 &= - 3 \left[ 2 \left( N^2 -1 \right) + N^2 \sum_{I = 1}^{4} \left( \fn_I - 1 \right) \left( \frac{\Delta_I}{\pi} - 1 \right)^2 \right] \, ,
\eea
while the gravitational anomaly $k$ reads
\be
 k = - \Tr  \gamma_3 = - 2 \left( N^2 -1 \right) - N^2 \sum_{I = 1}^{4} \left( \fn_I - 1 \right) = 2 \, .
\ee

The extremization of $c_{r} \left( \Delta_I , \fn_I \right)$ with respect to the $\Delta_I$ reproduces the exact central charge of the 2d CFT  \cite{Benini:2012cz,Benini:2013cda}. Notice that all the non-anomalous symmetries, including the baryonic one, enter in
the formula \eqref{c2d:anomalyconifold}, which depends on three independent fluxes and three independent fugacities. As pointed out in
 \cite{Benini:2015bwz}, the inclusion of baryonic charges is crucial when performing $c$-extremization. 

\section{High-temperature limit of the index for a generic theory}                                         
\label{high-temp limit of the index}

We can easily generalize the previous discussion to the case of general four-dimensional $\cN = 1$ SCFTs.
Our goal is to compute the partition function of $\cN = 1$ gauge theories
on $S^2 \times T^2$ with  a partial topological $A$-twist along $S^2$.
We identify, as before, the modulus of the torus with the \emph{fictitious} inverse temperature $\beta$, and we are interested in the \emph{high-temperature} limit $(\beta \to 0)$ of the index.
As we take $\beta$ to zero, we can use the asymptotic expansions \eqref{dedekind:hight:S}
and \eqref{theta:hight:S} for the elliptic functions appearing in the supersymmetric path integral \eqref{path integral index}. 
We focus on quiver gauge theories with bi-fundamental and adjoint chiral multiplets and a number $|G|$ of $\SU(N)^{(a)}$ gauge groups.
Eigenvalues $u_i^{(a)}$ and  gauge magnetic fluxes $\fm_i^{(a)}$ have to satisfy the tracelessness condition, \ie\;
\bea
\label{tracelessness condition}
\sum_{i = 1}^{N} u_i^{(a)} = 0 \, , \qquad \qquad \sum_{i = 1}^{N} \fm_i^{(a)} = 0 \, .
\eea
The magnetic fluxes and the chemical potentials for the global symmetries of the theory fulfill the constraints \eqref{supconstraints} and \eqref{supconstraints2}.
We also assume that $0<\Delta_I<2 \pi$. 

As in the previous examples, the solution to the BAEs is given by
\be
\label{BAEs:sol:SU(N)}
u_i^{(a)} = \cO \left( \beta \right) \, , \qquad \forall \quad i, a \, ,
\ee
and exists (up to discrete symmetries) only for  $\sum_{I \in W} \Delta_I  = 2 \pi$, for each monomial term $W$ in the superpotential, as we checked in many examples.
Due to this constraint, $\Delta_I / \pi$ behaves at all effects like a trial R-symmetry of the theory.

\subsection{Bethe potential at high temperature}
\label{Bethe potential at high-temp}

In this section we give the general rules for constructing the \emph{high-temperature} ``on-shell" Bethe potential for $\cN = 1$ quiver gauge theories to leading order in $1 / \beta$:
\begin{enumerate}
 \item A bi-fundamental field with chemical potential $\Delta_{(a,b)}$ transforming in the $({\bf N},\overline{\bf N})$ representation of $\SU(N)_a \times \SU(N)_b$, contributes
 \begin{equation}
 \label{Bethe potential bi-fundamental}
  \frac{i N^2}{\beta} F \left( \Delta_{(a,b)} \right) \, ,
 \end{equation}
 where the function $F$ is defined in (\ref{F:function}).
 \item An adjoint field with chemical potential $\Delta_{(a,a)}$ contributes
 \begin{equation}
 \label{Bethe potential  adjoint}
 \frac{i \left( N^2 - 1 \right)}{\beta} F \left( \Delta_{(a,a)} \right)  \, .
 \end{equation}
\end{enumerate}

\subsection{The topologically twisted index at high temperature}
\label{The index at high-temp}

Using the dominant solution \eqref{BAEs:sol:SU(N)} to the BAEs we can proceed to compute the topologically twisted index.
Here are the rules for constructing the logarithm of the index at high temperature to leading order in $1 / \beta$:
\begin{enumerate}
 \item For each group $a$, the contribution of the off-diagonal vector multiplet is
 \begin{equation}
  \label{index off-diag vector}
  - \frac{\left(N^2 - 1 \right)}{\beta} \frac{\pi^2}{3} \, .
 \end{equation}
 \item A bi-fundamental field with magnetic flux $\fn_{(a,b)}$ and chemical potential $\Delta_{(a,b)}$ transforming
 in the $({\bf N},\overline{\bf N})$ representation of $\SU(N)_a \times \SU(N)_b$, contributes
 \begin{equation}
 \label{index bi-fundamental}
 - \frac{N^2}{\beta} \left(\fn_{(a,b)} - 1 \right) F' \left( \Delta_{(a,b)} \right) \, .
 \end{equation}
 \item An adjoint field with magnetic flux $\fn_{(a,a)}$ and chemical potential $\Delta_{(a,a)}$, contributes
 \begin{equation}
 \label{index adjoint}
 - \frac{N^2 - 1}{\beta} \left(\fn_{(a,a)} - 1 \right) F' \left( \Delta_{(a,a)} \right) \, .
 \end{equation}    
\end{enumerate}

\subsection{An index theorem for the twisted matrix model}
\label{4d index theorem}

The high-temperature behavior of the index, to leading order in $1 / \beta$ and $N$,
can be captured by a simple universal formula involving the Bethe potential and its derivatives.
Let us recall some of the essential ingredients that we need in the following.

The R-symmetry 't Hooft anomaly of UV four-dimensional $\cN = 1$ SCFTs is given by
\bea\label{tHoof:linear:cubic:anomalies}
\Tr R^{\alpha} (\Delta_I) & = |G| \text{ dim } \SU(N) + \sum _{I} \text{ dim }\fR_I \left( \frac{\Delta_I}{\pi} - 1 \right)^{\alpha} \, ,
\eea
where the trace is taken over all the bi-fundamental fermions and gauginos
and $\text{dim }\fR_I$ is the dimension of the respective matter representation with R-charge $\Delta_I / \pi$.
On the other hand, the trial right-moving central charge of the IR two-dimensional $\cN = (0, 2)$ SCFT on $T^2$ can be computed from the spectrum of massless fermions \cite{Benini:2012cz,Benini:2013cda,Benini:2015bwz}.
These are gauginos with chirality $\gamma_3=1$ for all the gauge groups and  fermionic zero modes for each chiral field,
with a difference between the number  of fermions of opposite chiralities equal to $\fn_I-1$.
The central charge is related by the $\cN = 2$ superconformal algebra to the R-symmetry anomaly \cite{Benini:2012cz,Benini:2013cda},
and is given by
\bea
 \label{c2d:anomaly0}
 c_{r} \left( \Delta_I , \fn_I \right) & = - 3 \Tr  \gamma_3 R^2 \left( \Delta_I \right) \\
 & = - 3 \left[ |G| \text{ dim }\SU(N) + \sum_{I} \text{ dim }\fR_I \left( \fn_I - 1 \right) \left( \frac{\Delta_I}{\pi} - 1 \right)^2 \right]  \, .
\eea
By an explicit calculation we see that Eq.\,\eqref{c2d:anomaly0} can be rewritten as
\bea
\label{QFTrelation}
 c_{r} \left( \Delta_I, \fn_I \right) & = - 3 \Tr R^3 \left( \Delta_I \right) - \pi \sum_{I} \left[ \left( \fn_I - \frac{\Delta_I}{\pi} \right) \frac{\partial \Tr R^3 \left( \Delta_I \right)}{\partial \Delta_I} \right]
\, ,
\eea
where we used the relation \eqref{tHoof:linear:cubic:anomalies}.\footnote{Notice that, in evaluating the right hand side of  \eqref{QFTrelation}, we can consider all the $\Delta_I$  as independent variables and impose the constraints $\sum_{I \in W} \Delta_I  = 2 \pi$ only after differentiation. This is due to the form of the differential operator in \eqref{QFTrelation} and the constraints $\sum_{I \in W} \fn_I  = 2$.}
Moreover, the trial left-moving central charge of the 2d $\cN = (0, 2)$ theory reads
\be
\label{c2d:left:right:gr}
c_l = c_r - k \, , 
\ee
where $k$ is the gravitational anomaly and is given by
\be
\label{gr:anomaly:theorem}
 k = - \Tr \gamma_3 = - |G| \text{ dim }\SU(N) - \sum_{I} \text{ dim }\fR_I \left( \fn_I - 1 \right) \, .
\ee

For theories of D3-branes with an AdS dual, to leading order in $N$, the linear R-symmetry 't Hooft anomaly of the 4d theory vanishes, \ie\;$\Tr R= \cO(1)$ and $a=c$ \cite{Henningson:1998gx}.
Using the parameterization of a trial R-symmetry in terms of $\Delta_I / \pi$,  this is equivalent to
\be
\label{index theorem:constraint0}
\pi |G| + \sum_{I} \left( \Delta_I - \pi \right) = 0 \, ,
\ee
where the sum is taken over all matter fields (bi-fundamental and adjoint) in the quiver.
Similarly,  we have
\be
\label{index theorem:constraint1}
|G| + \sum_{I} \left( \fn_I - 1 \right) = 0 \, .
\ee
This is simply $k = - \Tr \gamma_3 = \cO(1)$, to leading order in $N$.

The {\it index theorem} can be expressed as:

\begin{theorem}
\label{theorem:1}
The topologically twisted index of any $\cN = 1$ $\SU(N)$ quiver gauge theory placed on  $S^2 \times T^2$ to leading order in $1/\beta$ is given  by
\be
 \label{index theorem:2d central charge}
 \log Z \left( \Delta_I, \fn_I \right) = \frac{\pi^2}{6 \beta} c_{l} \left( \Delta_I, \fn_I \right) \, ,
\ee
which is Cardy's universal formula for the asymptotic density of states in a CFT$_2$ \cite{Cardy:1986ie}. We can write the  extremal value of  the Bethe potential $\overline \cV \left( \Delta_I \right)$   as
\bea
 \label{index theorem: Bethe potential:finiteN}
 \overline\cV \left( \Delta_I \right) \equiv - i \cV \left( \Delta_I \right) \big|_{\text{BAEs}}
 & = \frac{\pi^3}{6 \beta} \left[ \Tr R^3 \left( \Delta_I \right) - \Tr R \left( \Delta_I \right) \right] \\
 & =  \frac{16 \pi^3}{27 \beta} \left[ 3 c \left( \Delta_I \right) - 2 a ( \Delta_I) \right] \, .
\eea
For theories of D3-branes at large $N$, the index can be recast as 
\bea
 \label{index theorem:attractor}
 \log Z \left( \Delta_I, \fn_I \right) & = - \frac{3}{\pi} \overline\cV \left( \Delta_I \right) - \sum_{I} \left[ \left( \fn_I - \frac{\Delta_I}{\pi} \right) \frac{\partial \overline\cV \left( \Delta_I \right)}{\partial \Delta_I} \right]
 \\ & = \frac{\pi^2}{6 \beta} c_{r} \left( \Delta_I, \fn_I \right) \, ,
\eea
where $\overline \cV \left( \Delta_I \right)$ reads
\bea
 \label{index theorem: Bethe potential:largeN}
 \overline\cV \left( \Delta_I \right) \equiv - i \cV \left( \Delta_I \right) \big|_{\text{BAEs}} = \frac{16 \pi^3}{27 \beta} a ( \Delta_I) \, .
\eea
\end{theorem}
\begin{proof} Observe first that again we can consider all the $\Delta_I$ in \eqref{index theorem:attractor} as independent variables and impose the constraints $\sum_{I \in W} \Delta_I  = 2 \pi$ only after differentiation. This is due to the form of the differential operator in \eqref{index theorem:attractor} and $\sum_{I \in W} \fn_I  = 2$.
To prove the first equality in \eqref{index theorem:attractor}, we promote the explicit factors of $\pi$, appearing in \eqref{Bethe potential bi-fundamental} and \eqref{Bethe potential adjoint}, to a formal variable $\bm{\pi}$.
Notice that the ``on-shell'' Bethe potential $\wb{\mathcal{V}}$, at large $N$, is a homogeneous function of $\Delta_I$ and $\bm{\pi}$, \ie\;
\begin{equation}
  \overline{\mathcal{V}}(\lambda \Delta_I, \lambda \bm{\pi}) = \lambda^3 \, \wb{\mathcal{V}}(\Delta_I, \bm{\pi}) \, .
\end{equation}
Hence,
\begin{equation}\label{hom}
 \frac{\partial \wb{\mathcal{V}}(\Delta_I, \bm{\pi})}{\partial \bm{\pi}} =
 \frac{1}{\bm{\pi}} \left[ 3 \, \wb{\mathcal{V}}(\Delta_I) -\sum_I  \Delta_I \frac{\partial \wb{\mathcal{V}}(\Delta_I)}{\partial \Delta_I} \right]\, .
\end{equation}
Now, we consider a generic quiver gauge theory with matters in bi-fundamental and adjoint representations of the gauge group.
They contribute to the Bethe potential $\overline \cV (\Delta_I, \bm{\pi})$ as written in \eqref{Bethe potential bi-fundamental} and \eqref{Bethe potential adjoint}, respectively.
Let us calculate the derivative of $\overline \cV (\Delta_I, \bm{\pi})$ with respect to $\Delta_{I}$:
 \be
 \label{proof bf:Delta}
 - \sum_{I} \fn_I \frac{\partial \overline\cV(\Delta_I, \bm{\pi})}{\partial \Delta_I} =
 - \frac{N^2}{\beta} \sum_{I} \fn_I\; F' \left( \Delta_{I} \right) \, .
 \ee
Next, we take the derivative of the Bethe potential with respect to $\bm{\pi}$:
\begin{align}
 \label{proof bf:pi}
 - \sum_{I} \frac{\partial \overline \cV(\Delta_I, \bm{\pi})}{\partial \bm{\pi}} =
 \frac{N^2}{\beta} \sum_{I} \; F' \left( \Delta_{I} \right)
 - \frac{N^2}{\beta} \sum_{I} \left( \frac{\bm{\pi}^2}{3} - \frac{\bm{\pi}}{3} \Delta_{I} \right) \, .
\end{align}
Using \eqref{hom} and combining \eqref{proof bf:Delta} with the first term of \eqref{proof bf:pi} as in the right hand side of Eq.\,\eqref{index theorem:attractor}, 
we obtain the contribution of matter fields \eqref{index bi-fundamental} and \eqref{index adjoint} to the index.
The contribution of the second term in \eqref{proof bf:pi} to  Eq.\,\eqref{index theorem:attractor} can be written as
\be
 - \frac{N^2}{\beta} \frac{\bm{\pi}}{3} \sum_I \left( \bm{\pi} - \Delta_I \right) =
 - \frac{N^2}{\beta} \frac{\bm{\pi}^2}{3} |G| \, ,
\ee
where we used the constraint \eqref{index theorem:constraint0}.
This is precisely the contribution of the off-diagonal vector multiplets \eqref{index off-diag vector} to the index at large $N$.

Parameterizing the trial R-symmetry of an $\cN = 1$ theory in terms of $\Delta_I / \pi$, we can prove \eqref{index theorem: Bethe potential:finiteN}:
\bea
 \overline \cV \left( \Delta_I \right)
 & = \frac{1}{\beta} \sum_{I } \text{ dim } \fR_I \; F \left( \Delta_I \right)
 = \frac{1}{6 \beta} \sum_{I} \text{ dim } \fR_I \left[ \left(\Delta_I - \pi \right)^3 - \pi^2 \left( \Delta_I - \pi \right) \right] \\
 & = \frac{\pi^3}{6 \beta} \left[\sum _{I} \text{ dim } \fR_I \left( \frac{\Delta_I}{\pi} - 1 \right)^3
 - \sum _{I} \text{ dim } \fR_I \left( \frac{\Delta_I}{\pi} - 1 \right)\right] \\
 & = \frac{\pi^3}{6 \beta} \left[ \Tr R^3 \left( \Delta_I \right) - \Tr R \left( \Delta_I \right) \right] \, ,
\eea
which at large $N$, due to \eqref{index theorem:constraint0}, is equal to \eqref{index theorem: Bethe potential:largeN}.

Finally, we need to show that the high-temperature limit of the index is given by the Cardy formula \eqref{index theorem:2d central charge}.
Bi-fundamental and adjoint fields contribute to the index according to \eqref{index bi-fundamental} and \eqref{index adjoint}, respectively.
We thus have
\bea
 \log Z \left( \Delta_I, \fn_I \right) & = - \frac{1}{\beta} \left[ \frac{\pi^2}{3} |G| \text{ dim }\SU(N) + \sum_{I} \text{ dim }\fR_I \; \left( \fn_I - 1 \right) F' \left( \Delta_I \right) \right] \\
 & = - \frac{\pi^2}{6 \beta} \left\{ 2 |G| \text{ dim }\SU(N) + \sum_{I} \text{ dim }\fR_I \; \left( \fn_I - 1 \right) \left[ 3 \left( \frac{\Delta_I}{\pi} - 1 \right)^2 - 1 \right] \right\} \\
 & = \frac{\pi^2}{6 \beta} \left[ c_{r} \left( \Delta_I, \fn_I \right) + \Tr \gamma_3 \right]
 = \frac{\pi^2}{6 \beta} c_l \left( \Delta_I, \fn_I \right) \, ,
\eea
where we used \eqref{c2d:anomaly0} and \eqref{c2d:left:right:gr} in the third and the fourth equality, respectively.
For quiver gauge theories fulfilling the constraint \eqref{index theorem:constraint1} the above formula reduces to the second equality in \eqref{index theorem:attractor} at large $N$. This completes the proof.
\end{proof}

It is worth stressing  that, when using formula \eqref{index theorem:attractor}, the linear  relations among the $\Delta_I$ can be imposed after differentiation.
It is always possible, ignoring some linear relations, to parameterize $\overline \cV (\Delta_I)$ such that it becomes a homogeneous function of degree 3 in the chemical potentials $\Delta_I$ \cite{Benvenuti:2006xg}.
With this parameterization the index theorem becomes
\be
 \log Z \left( \Delta_I, \fn_I \right) = - \sum_{I} \fn_I \frac{\partial \overline\cV \left( \Delta_I \right)}{\partial \Delta_I} \, .
\ee
As we have seen, this is indeed the case for  $\cN = 4$ SYM and the Klebanov-Witten theory.
We note that our result is  very similar to that obtained for the large $N$ limit of the topologically twisted index
of three-dimensional $\cN \geq 2$ Yang-Mills-Chern-Simons-matter theories placed on $S^2 \times S^1$ \cite{Hosseini:2016tor,Hosseini:2016ume}.

\section{Future directions}
\label{discussion}

There are various directions to explore. Let us mention some of them.
\begin{enumerate}
 \item  We can refine the index by turning on angular momentum along the two-dimensional compact manifold $S^2$ \cite{Benini:2015noa}.
 It would be quite interesting to understand the results for the refined index in the context of  rotating black string solutions in five-dimensional gauged supergravity
 which are still to be found.
 
 \item The critical points of the Bethe potential $\cV (\{u_i^{(a)}\})$ coincide with the Bethe equations for the vacua of a quantum integrable system \cite{Moore:1997dj,Gerasimov:2006zt,Gerasimov:2007ap,Nekrasov:2009ui,Nekrasov:2009rc,Nekrasov:2014xaa}.\footnote{See \cite{Beem:2012mb,Nieri:2015yia} for a discussion about the Bethe equations in the context of factorization and holomorphic blocks.}
 It would be very interesting to understand if the quantum integrability picture can shed new light on the microscopic origin of black holes/strings entropy. 
 
 \item Regular asymptotically AdS$_5\times S^5$ rotating black holes, characterized by three electric charges and two angular momenta, have been found in five-dimensional $\U(1)^3$ gauged supergravity \cite{Gutowski:2004ez,Gutowski:2004yv,Kunduri:2006ek,Chong:2005hr}.
 Our general results for 4d/5d static black holes/strings may suggest new approaches for understanding  the statistical meaning of the Bekenstein-Hawking entropy
 for this class of black holes in terms of states in the dual $\cN = 4$ SYM theory.
\end{enumerate}
We hope to come back to these questions  soon.

\section*{Acknowledgments}

We would like to thank Francesco Benini, Nikolay Bobev, Luca Cassia, Kiril Hristov and Sara Pasquetti  for useful discussions.
We are also grateful to Nikolay Bobev for commenting on a draft of this paper.
AN and AZ are supported by the INFN and the MIUR-FIRB grant RBFR10QS5J ``String Theory and Fundamental Interactions''.
SMH is supported in part by INFN. AZ  thank the Galileo Galilei Institute for Theoretical Physics for the hospitality and the INFN for partial support during the completion of this work.
\begin{appendix}

\section{Elliptic functions and their asymptotics}
\label{Elliptic functions}

The Dedekind eta function is defined by
\be
\eta(q) = \eta(\tau) = q^{\frac{1}{24}} \prod_{n=1}^{\infty} \left(1 - q^{n} \right) \, , \qquad \qquad \im \tau > 0 \, ,
\label{dedekind:eta}
\ee
where $q = e^{2 \pi i \tau}$.
It has the following modular properties
\be\label{modular_eta}
\eta(\tau + 1) = e^{\frac{i \pi }{12}} \, \eta (\tau)\, , \qquad \qquad \eta\left( - \frac{1}{\tau} \right) = \sqrt{- i \tau} \, \eta (\tau) \, .
\ee
The Jacobi theta function reads
\begin{align}
\theta_1 (x;q) = \theta_1 (u;\tau) & = - i q^{\frac18} x^{\frac12} \prod_{k=1}^{\infty} \left( 1- q^k \right) \left( 1-x q^k \right) \left( 1- x^{-1} q^{k-1} \right) \nn \\
& = -i \sum_{n \in \mathbb{Z}} (-1)^n e^{i u \left( n+ \frac12 \right)} e^{\pi i \tau \left( n+ \frac12 \right)^2}\, ,
\label{theta:function}
\end{align}
where $x = e^{i u}$ and $q$ is as before.
The function $\theta_1 (u; \tau)$ has simple zeros in $u$ at $u = 2 \pi \mathbb{Z} + 2 \pi \tau \mathbb{Z}$ and no poles.
Its modular properties are,
\be\label{modular_theta}
\theta_1\left(u;\tau+1\right) = e^{\frac{i \pi}{4}} \, \theta_1\left(u;\tau\right)\, , \qquad \qquad
\theta_1\left( \frac{u}{\tau};-\frac{1}{\tau}\right) = - i \sqrt{- i \tau} \, e^{\frac{i u^2}{4 \pi \tau}} \, \theta_1\left( u; \tau\right) \, .
\ee
We also note the following useful formula,
\be
\label{theta:function:shift}
\theta_1\left( q^m x;q\right) = (-1)^{-m}\,x^{-m} q^{-\frac{m^2}{2}}\theta_1(x;q) \, , \qquad \qquad m \in \bZ \, .
\ee

The asymptotic behavior of the $\eta(q)$ and $\theta_1(x; q)$ as $q \to 1$ can be derived by using their modular properties.
To this purpose, we first need to perform an $S$-transformation, \ie\;$\tau\to -1/\tau$, and then expand the
resulting functions in series of $q$, which is now a small parameter in the $\tau\to i0$ limit. 

Let us start with the Dedekind $\eta$-function.
The action of modular transformation is written in (\ref{modular_eta}). 
We will identify the ``inverse temperature'' $\beta$ with the modular parameter $\tau$ of the torus: $\tau=i\beta/2\pi$.
Then, expanding the $S$-transformed $\eta$-function we get 
\begin{align}
 \label{dedekind:hight:S}
  \log\left[\eta(\tau)\right] & = -\frac{1}{2}\log\left( -i \tau \right) + \log\left[ \eta\left( -\frac{1}{\tau}\right)\right] \nn \\
  & = - \frac{1}{2}\log\left( \frac{\beta}{2\pi}\right)-\frac{\pi^2}{6\beta} + \cO\left( e^{- 1 / \beta} \right) \, .
\end{align}
Similarly, we can consider the asymptotic expansion of the Jacobi $\theta$-function:
\bea
 \log\left[\theta_1(u;\tau) \right] & =
 \frac{i \pi}{2} - \frac{1}{2}\log\left( -i \tau \right) - \frac{i u^2 }{4 \pi  \tau}
 + \log\left[ \theta_1\left(\frac{u}{\tau}; -\frac{1}{\tau}\right)\right] \nn \\
 & = -\frac{\pi^2}{2\beta} - \frac{u^2}{2\beta}-\frac{1}{2} \log\left( \frac{\beta}{2\pi} \right)
 + \log\left[2 \sinh\left(\frac{\pi u}{\beta} \right)\right]
 + \cO \left( e^{- 1 / \beta} \right) \, ,
\eea
Writing $2\sinh\left(\frac{\pi u}{\beta} \right) = e^{\pi u/\beta}\left( 1 - e^{- 2 \pi u/\beta} \right)$,
we have the following expansion
\be
 \log\left[2\sinh\left(\frac{\pi u}{\beta} \right)\right]
 = \frac{\pi}{\beta} u \sign\left[ \re (u) \right]
 - \sum_{k=1}^{\infty} \frac{1}{k} e^{- \frac{2 k \pi}{\beta} u \sign\left[ \re (u) \right]} \, .
\ee
Putting all pieces together, we find
\bea
\label{theta:hight:S}
\log\left[ \theta_1(u;\tau) \right] = -\frac{\pi^2}{2\beta} - \frac{u^2}{2\beta}-\frac{1}{2} \log\left( \frac{\beta}{2\pi} \right)
+ \frac{\pi}{\beta} u \sign\left[ \re(u) \right]
+ \cO \left( e^{- 1 / \beta} \right) \, .
\eea

\section{Anomaly cancellation}
\label{Anomaly cancellation}

Here we obtain the conditions for which the integrand in \eqref{path integral index} is a well-defined meromorphic function on the torus.
To this aim the one-loop contributions must be invariant under the transformation $x^{\rho} \to q^{\rho(\gamma)}\, x^{\rho}$ where $\gamma$ live in the co-root lattice $\Gamma_\fh$ of the gauge group. 

The off-diagonal vector multiplets contribute to the index as
\be
Z_{1-\text{loop}}^{\text{gauge, off}}=(-1)^{\sum_{\alpha>0}\alpha(\fm)}\prod_{\alpha\in G}\left[ \frac{\theta_1(x^\alpha ; q)}{i \eta(q)} \right] \, .
\ee
Applying $x^{\rho} \to q^{\rho(\gamma)}\, x^{\rho}$ and using Eq.\,\eqref{theta:function:shift} we find
\be
\label{elliptic:transform:vec}
Z_{1-\text{loop}}^{\text{gauge, off}} \to Z_{1-\text{loop}}^{\text{gauge, off}} \prod_{\alpha \in G} (-1)^{-\alpha(\gamma)} \, e^{- i \pi \tau \alpha(\gamma)^2} \, e^{-i \alpha(u) \alpha(\gamma)} \, .
\ee
The one-loop contribution of a chiral multiplet is given by
\be
Z_{1-\text{loop}}^{\text{chiral}} = \prod_{\rho_I \in \fR_{I}} \bigg[ \frac{i \eta(q)}{\theta_1(x^{\rho_I} y^{\nu_I} ; q)} \bigg]^{B} \, , \qquad \qquad B = \rho_I(\fm)- \fn_I + 1 \, ,
\ee
where $\nu_I$ is the weight of the multiplet under the flavor symmetry group.
It transforms as
\be
\label{elliptic:transform:chiral}
Z_{1-\text{loop}}^{\text{chiral}} \to Z_{1-\text{loop}}^{\text{chiral}} \prod_{\rho_I \in \fR_{I}} (-1)^{\rho_I(\gamma) B} \, e^{ i \pi \tau \rho_I(\gamma)^2 B} \, e^{i \rho_I(u) \rho_I(\gamma) B} \, e^{ i \rho_I(\gamma) \nu_I(\Delta) B} \, .
\ee
Putting everything together, the total prefactor in the integrand vanishes if we demand the following anomaly cancellation conditions:
\begin{align}
&\sum_{\alpha\in G} \alpha(\gamma)^2 + \sum_{I} \sum_{\rho_I\in\fR_{I}} \left(\fn_I - 1\right) \rho_I(\gamma)^2 = 0\, ,\hspace{2.25cm} \mbox{$\U(1)_R$-gauge-gauge anomaly}\, ,\\
&\sum_{\alpha\in G} \alpha(\gamma) \alpha(u) + \sum_{I} \sum_{\rho_I\in\fR_{I}} \left(\fn_I - 1\right) \rho_I(\gamma) \rho(u) = 0\, ,\hspace{1.cm} \mbox{$\U(1)_R$-gauge-gauge anomaly}\, ,\\
&\sum_I\sum_{\rho_I \in \fR_{I}} \rho_I(\gamma)^2 \, \rho_I(\fm) = 0\, , \hspace{4.9cm} \mbox{gauge$^3$ anomaly}\, ,\\
&\sum_I\sum_{\rho_I \in \fR_{I}} \rho_I(\gamma) \, \rho(u) \, \rho_I(\fm) = 0\, , \hspace{4.2cm} \mbox{gauge$^3$ anomaly}\, ,\\
&\sum_I\sum_{\rho_I \in \fR_{I}} \rho_I(\gamma) \, \rho_I(\fm) \, \nu_I(\Delta) = 0\, , \hspace{3.95cm} \mbox{gauge-gauge-flavor anomaly}\, ,\\
&\sum_I\sum_{\rho_I \in \fR_{I}} \left(\fn_I - 1\right) \rho_I(\gamma) \, \nu_I(\Delta) = 0\, , \hspace{3.525cm} \mbox{$\U(1)_R$-gauge-flavor anomaly} \, .
\end{align}
The signs cancel out automatically for all D3-brane quivers since the number of arrows entering a node equals the number of arrows leaving it.

\end{appendix}

\bibliographystyle{ytphys}

\bibliography{index}

\providecommand{\href}[2]{#2}\begingroup\raggedright\begin{thebibliography}{10}

\bibitem{Benini:2015noa}
F.~Benini and A.~Zaffaroni, ``{A topologically twisted index for
  three-dimensional supersymmetric theories},''
  \href{http://dx.doi.org/10.1007/JHEP07(2015)127}{{\em JHEP} {\bfseries 07}
  (2015) 127},
\href{http://arxiv.org/abs/1504.03698}{{\ttfamily arXiv:1504.03698 [hep-th]}}.

\bibitem{Closset:2013sxa}
C.~Closset and I.~Shamir, ``{The $\mathcal{N}=1$ Chiral Multiplet on $T^2\times
  S^2$ and Supersymmetric Localization},''
  \href{http://dx.doi.org/10.1007/JHEP03(2014)040}{{\em JHEP} {\bfseries 03}
  (2014) 040},
\href{http://arxiv.org/abs/1311.2430}{{\ttfamily arXiv:1311.2430 [hep-th]}}.

\bibitem{Honda:2015yha}
M.~Honda and Y.~Yoshida, ``{Supersymmetric index on $T^2 x S^2$ and elliptic
  genus},''
\href{http://arxiv.org/abs/1504.04355}{{\ttfamily arXiv:1504.04355 [hep-th]}}.

\bibitem{Closset:2015rna}
C.~Closset, S.~Cremonesi, and D.~S. Park, ``{The equivariant A-twist and gauged
  linear sigma models on the two-sphere},''
  \href{http://dx.doi.org/10.1007/JHEP06(2015)076}{{\em JHEP} {\bfseries 06}
  (2015) 076},
\href{http://arxiv.org/abs/1504.06308}{{\ttfamily arXiv:1504.06308 [hep-th]}}.

\bibitem{Benini:2016hjo}
F.~Benini and A.~Zaffaroni, ``{Supersymmetric partition functions on Riemann
  surfaces},''
\href{http://arxiv.org/abs/1605.06120}{{\ttfamily arXiv:1605.06120 [hep-th]}}.

\bibitem{Closset:2016arn}
C.~Closset and H.~Kim, ``{Comments on twisted indices in 3d supersymmetric
  gauge theories},'' \href{http://dx.doi.org/10.1007/JHEP08(2016)059}{{\em
  JHEP} {\bfseries 08} (2016) 059},
\href{http://arxiv.org/abs/1605.06531}{{\ttfamily arXiv:1605.06531 [hep-th]}}.

\bibitem{Benini:2015eyy}
F.~Benini, K.~Hristov, and A.~Zaffaroni, ``{Black hole microstates in AdS$_4$
  from supersymmetric localization},''
\href{http://arxiv.org/abs/1511.04085}{{\ttfamily arXiv:1511.04085 [hep-th]}}.

\bibitem{Benini:2016rke}
F.~Benini, K.~Hristov, and A.~Zaffaroni, ``{Exact microstate counting for
  dyonic black holes in AdS4},''
\href{http://arxiv.org/abs/1608.07294}{{\ttfamily arXiv:1608.07294 [hep-th]}}.

\bibitem{Hosseini:2016tor}
S.~M. Hosseini and A.~Zaffaroni, ``{Large $N$ matrix models for 3d ${\cal N}=2$
  theories: twisted index, free energy and black holes},''
  \href{http://dx.doi.org/10.1007/JHEP08(2016)064}{{\em JHEP} {\bfseries 08}
  (2016) 064},
\href{http://arxiv.org/abs/1604.03122}{{\ttfamily arXiv:1604.03122 [hep-th]}}.

\bibitem{Hosseini:2016ume}
S.~M. Hosseini and N.~Mekareeya, ``{Large $N$ topologically twisted index:
  necklace quivers, dualities, and Sasaki-Einstein spaces},''
  \href{http://dx.doi.org/10.1007/JHEP08(2016)089}{{\em JHEP} {\bfseries 08}
  (2016) 089},
\href{http://arxiv.org/abs/1604.03397}{{\ttfamily arXiv:1604.03397 [hep-th]}}.

\bibitem{Benini:2012cz}
F.~Benini and N.~Bobev, ``{Exact two-dimensional superconformal R-symmetry and
  c-extremization},''
  \href{http://dx.doi.org/10.1103/PhysRevLett.110.061601}{{\em Phys. Rev.
  Lett.} {\bfseries 110} no.~6, (2013) 061601},
\href{http://arxiv.org/abs/1211.4030}{{\ttfamily arXiv:1211.4030 [hep-th]}}.

\bibitem{Benini:2013cda}
F.~Benini and N.~Bobev, ``{Two-dimensional SCFTs from wrapped branes and
  c-extremization},'' \href{http://dx.doi.org/10.1007/JHEP06(2013)005}{{\em
  JHEP} {\bfseries 06} (2013) 005},
\href{http://arxiv.org/abs/1302.4451}{{\ttfamily arXiv:1302.4451 [hep-th]}}.

\bibitem{Benini:2015bwz}
F.~Benini, N.~Bobev, and P.~M. Crichigno, ``{Two-dimensional SCFTs from
  D3-branes},'' \href{http://dx.doi.org/10.1007/JHEP07(2016)020}{{\em JHEP}
  {\bfseries 07} (2016) 020},
\href{http://arxiv.org/abs/1511.09462}{{\ttfamily arXiv:1511.09462 [hep-th]}}.

\bibitem{Bobev:2014jva}
N.~Bobev, K.~Pilch, and O.~Vasilakis, ``{(0, 2) SCFTs from the Leigh-Strassler
  fixed point},'' \href{http://dx.doi.org/10.1007/JHEP06(2014)094}{{\em JHEP}
  {\bfseries 06} (2014) 094},
\href{http://arxiv.org/abs/1403.7131}{{\ttfamily arXiv:1403.7131 [hep-th]}}.

\bibitem{Benini:2013xpa}
F.~Benini, R.~Eager, K.~Hori, and Y.~Tachikawa, ``{Elliptic Genera of 2d
  ${\mathcal{N}}$ = 2 Gauge Theories},''
  \href{http://dx.doi.org/10.1007/s00220-014-2210-y}{{\em Commun. Math. Phys.}
  {\bfseries 333} no.~3, (2015) 1241--1286},
\href{http://arxiv.org/abs/1308.4896}{{\ttfamily arXiv:1308.4896 [hep-th]}}.

\bibitem{Benini:2013nda}
F.~Benini, R.~Eager, K.~Hori, and Y.~Tachikawa, ``{Elliptic genera of
  two-dimensional N=2 gauge theories with rank-one gauge groups},''
  \href{http://dx.doi.org/10.1007/s11005-013-0673-y}{{\em Lett. Math. Phys.}
  {\bfseries 104} (2014) 465--493},
\href{http://arxiv.org/abs/1305.0533}{{\ttfamily arXiv:1305.0533 [hep-th]}}.

\bibitem{Gadde:2015wta}
A.~Gadde, S.~S. Razamat, and B.~Willett, ``{On the reduction of 4d $
  \mathcal{N}=1 $ theories on $ {\mathbb{S}}^2 $},''
  \href{http://dx.doi.org/10.1007/JHEP11(2015)163}{{\em JHEP} {\bfseries 11}
  (2015) 163},
\href{http://arxiv.org/abs/1506.08795}{{\ttfamily arXiv:1506.08795 [hep-th]}}.

\bibitem{Nekrasov:2014xaa}
N.~A. Nekrasov and S.~L. Shatashvili, ``{Bethe/Gauge correspondence on curved
  spaces},'' \href{http://dx.doi.org/10.1007/JHEP01(2015)100}{{\em JHEP}
  {\bfseries 01} (2015) 100},
\href{http://arxiv.org/abs/1405.6046}{{\ttfamily arXiv:1405.6046 [hep-th]}}.

\bibitem{Yang1969}
C.~N. Yang and C.~P. Yang, ``Thermodynamics of a one-dimensional system of
  bosons with repulsive delta-function interaction,'' {\em Journal of
  Mathematical Physics} {\bfseries 10} no.~7, (1969) .

\bibitem{Anselmi:1997am}
D.~Anselmi, D.~Z. Freedman, M.~T. Grisaru, and A.~A. Johansen,
  ``{Nonperturbative formulas for central functions of supersymmetric gauge
  theories},'' \href{http://dx.doi.org/10.1016/S0550-3213(98)00278-8}{{\em
  Nucl. Phys.} {\bfseries B526} (1998) 543--571},
\href{http://arxiv.org/abs/hep-th/9708042}{{\ttfamily arXiv:hep-th/9708042
  [hep-th]}}.

\bibitem{Henningson:1998gx}
M.~Henningson and K.~Skenderis, ``{The Holographic Weyl anomaly},''
  \href{http://dx.doi.org/10.1088/1126-6708/1998/07/023}{{\em JHEP} {\bfseries
  07} (1998) 023},
\href{http://arxiv.org/abs/hep-th/9806087}{{\ttfamily arXiv:hep-th/9806087
  [hep-th]}}.

\bibitem{Intriligator:2003jj}
K.~A. Intriligator and B.~Wecht, ``{The Exact superconformal R symmetry
  maximizes a},'' \href{http://dx.doi.org/10.1016/S0550-3213(03)00459-0}{{\em
  Nucl. Phys.} {\bfseries B667} (2003) 183--200},
\href{http://arxiv.org/abs/hep-th/0304128}{{\ttfamily arXiv:hep-th/0304128
  [hep-th]}}.

\bibitem{Jafferis:2011zi}
D.~L. Jafferis, I.~R. Klebanov, S.~S. Pufu, and B.~R. Safdi, ``{Towards the
  F-Theorem: N=2 Field Theories on the Three-Sphere},''
  \href{http://dx.doi.org/10.1007/JHEP06(2011)102}{{\em JHEP} {\bfseries 06}
  (2011) 102},
\href{http://arxiv.org/abs/1103.1181}{{\ttfamily arXiv:1103.1181 [hep-th]}}.

\bibitem{Gubser:1998vd}
S.~S. Gubser, ``{Einstein manifolds and conformal field theories},''
  \href{http://dx.doi.org/10.1103/PhysRevD.59.025006}{{\em Phys. Rev.}
  {\bfseries D59} (1999) 025006},
\href{http://arxiv.org/abs/hep-th/9807164}{{\ttfamily arXiv:hep-th/9807164
  [hep-th]}}.

\bibitem{Butti:2005vn}
A.~Butti and A.~Zaffaroni, ``{R-charges from toric diagrams and the equivalence
  of a-maximization and Z-minimization},''
  \href{http://dx.doi.org/10.1088/1126-6708/2005/11/019}{{\em JHEP} {\bfseries
  11} (2005) 019},
\href{http://arxiv.org/abs/hep-th/0506232}{{\ttfamily arXiv:hep-th/0506232
  [hep-th]}}.

\bibitem{Martelli:2005tp}
D.~Martelli, J.~Sparks, and S.-T. Yau, ``{The Geometric dual of a-maximisation
  for Toric Sasaki-Einstein manifolds},''
  \href{http://dx.doi.org/10.1007/s00220-006-0087-0}{{\em Commun. Math. Phys.}
  {\bfseries 268} (2006) 39--65},
\href{http://arxiv.org/abs/hep-th/0503183}{{\ttfamily arXiv:hep-th/0503183
  [hep-th]}}.

\bibitem{Martelli:2006yb}
D.~Martelli, J.~Sparks, and S.-T. Yau, ``{Sasaki-Einstein manifolds and volume
  minimisation},'' \href{http://dx.doi.org/10.1007/s00220-008-0479-4}{{\em
  Commun. Math. Phys.} {\bfseries 280} (2008) 611--673},
\href{http://arxiv.org/abs/hep-th/0603021}{{\ttfamily arXiv:hep-th/0603021
  [hep-th]}}.

\bibitem{Ferrara:1996dd}
S.~Ferrara and R.~Kallosh, ``{Supersymmetry and attractors},''
  \href{http://dx.doi.org/10.1103/PhysRevD.54.1514}{{\em Phys. Rev.} {\bfseries
  D54} (1996) 1514--1524},
\href{http://arxiv.org/abs/hep-th/9602136}{{\ttfamily arXiv:hep-th/9602136
  [hep-th]}}.

\bibitem{Dall'Agata:2010gj}
G.~Dall'Agata and A.~Gnecchi, ``{Flow equations and attractors for black holes
  in $\mathcal{N}{=}2$ $U(1)$ gauged supergravity},''
  \href{http://dx.doi.org/10.1007/JHEP03(2011)037}{{\em JHEP} {\bfseries 03}
  (2011) 037},
\href{http://arxiv.org/abs/1012.3756}{{\ttfamily arXiv:1012.3756 [hep-th]}}.

\bibitem{Karndumri:2013iqa}
P.~Karndumri and E.~O~Colgain, ``{Supergravity dual of $c$-extremization},''
  \href{http://dx.doi.org/10.1103/PhysRevD.87.101902}{{\em Phys. Rev.}
  {\bfseries D87} no.~10, (2013) 101902},
\href{http://arxiv.org/abs/1302.6532}{{\ttfamily arXiv:1302.6532 [hep-th]}}.

\bibitem{Hristov:2014eza}
K.~Hristov, ``{Dimensional reduction of BPS attractors in AdS gauged
  supergravities},'' \href{http://dx.doi.org/10.1007/JHEP12(2014)066}{{\em
  JHEP} {\bfseries 12} (2014) 066},
\href{http://arxiv.org/abs/1409.8504}{{\ttfamily arXiv:1409.8504 [hep-th]}}.

\bibitem{Amariti:2016mnz}
A.~Amariti and C.~Toldo, ``{Betti multiplets, flows across dimensions and
  c-extremization},''
\href{http://arxiv.org/abs/1610.08858}{{\ttfamily arXiv:1610.08858 [hep-th]}}.

\bibitem{Klemm:2016kxw}
D.~Klemm, N.~Petri, and M.~Rabbiosi, ``{Black string first order flow in N=2,
  d=5 abelian gauged supergravity},''
\href{http://arxiv.org/abs/1610.07367}{{\ttfamily arXiv:1610.07367 [hep-th]}}.

\bibitem{Kawai:1993jk}
T.~Kawai, Y.~Yamada, and S.-K. Yang, ``{Elliptic genera and N=2 superconformal
  field theory},'' \href{http://dx.doi.org/10.1016/0550-3213(94)90428-6}{{\em
  Nucl. Phys.} {\bfseries B414} (1994) 191--212},
\href{http://arxiv.org/abs/hep-th/9306096}{{\ttfamily arXiv:hep-th/9306096
  [hep-th]}}.

\bibitem{Benjamin:2015hsa}
N.~Benjamin, M.~C.~N. Cheng, S.~Kachru, G.~W. Moore, and N.~M. Paquette,
  ``{Elliptic Genera and 3d Gravity},''
  \href{http://dx.doi.org/10.1007/s00023-016-0469-6}{{\em Annales Henri
  Poincare} {\bfseries 17} no.~10, (2016) 2623--2662},
\href{http://arxiv.org/abs/1503.04800}{{\ttfamily arXiv:1503.04800 [hep-th]}}.

\bibitem{DiPietro:2014bca}
L.~Di~Pietro and Z.~Komargodski, ``{Cardy formulae for SUSY theories in $d =$ 4
  and $d =$ 6},'' \href{http://dx.doi.org/10.1007/JHEP12(2014)031}{{\em JHEP}
  {\bfseries 12} (2014) 031},
\href{http://arxiv.org/abs/1407.6061}{{\ttfamily arXiv:1407.6061 [hep-th]}}.

\bibitem{Ardehali:2015hya}
A.~Arabi~Ardehali, J.~T. Liu, and P.~Szepietowski, ``{High-Temperature
  Expansion of Supersymmetric Partition Functions},''
  \href{http://dx.doi.org/10.1007/JHEP07(2015)113}{{\em JHEP} {\bfseries 07}
  (2015) 113},
\href{http://arxiv.org/abs/1502.07737}{{\ttfamily arXiv:1502.07737 [hep-th]}}.

\bibitem{Ardehali:2015bla}
A.~Arabi~Ardehali, ``{High-temperature asymptotics of supersymmetric partition
  functions},'' \href{http://dx.doi.org/10.1007/JHEP07(2016)025}{{\em JHEP}
  {\bfseries 07} (2016) 025},
\href{http://arxiv.org/abs/1512.03376}{{\ttfamily arXiv:1512.03376 [hep-th]}}.

\bibitem{Lorenzen:2014pna}
J.~Lorenzen and D.~Martelli, ``{Comments on the Casimir energy in
  supersymmetric field theories},''
  \href{http://dx.doi.org/10.1007/JHEP07(2015)001}{{\em JHEP} {\bfseries 07}
  (2015) 001},
\href{http://arxiv.org/abs/1412.7463}{{\ttfamily arXiv:1412.7463 [hep-th]}}.

\bibitem{Assel:2015nca}
B.~Assel, D.~Cassani, L.~Di~Pietro, Z.~Komargodski, J.~Lorenzen, and
  D.~Martelli, ``{The Casimir Energy in Curved Space and its Supersymmetric
  Counterpart},'' \href{http://dx.doi.org/10.1007/JHEP07(2015)043}{{\em JHEP}
  {\bfseries 07} (2015) 043},
\href{http://arxiv.org/abs/1503.05537}{{\ttfamily arXiv:1503.05537 [hep-th]}}.

\bibitem{Bobev:2015kza}
N.~Bobev, M.~Bullimore, and H.-C. Kim, ``{Supersymmetric Casimir Energy and the
  Anomaly Polynomial},'' \href{http://dx.doi.org/10.1007/JHEP09(2015)142}{{\em
  JHEP} {\bfseries 09} (2015) 142},
\href{http://arxiv.org/abs/1507.08553}{{\ttfamily arXiv:1507.08553 [hep-th]}}.

\bibitem{Genolini:2016sxe}
P.~B. Genolini, D.~Cassani, D.~Martelli, and J.~Sparks, ``{The holographic
  supersymmetric Casimir energy},''
\href{http://arxiv.org/abs/1606.02724}{{\ttfamily arXiv:1606.02724 [hep-th]}}.

\bibitem{DiPietro:2016ond}
L.~Di~Pietro and M.~Honda, ``{Cardy Formula for 4d SUSY Theories and
  Localization},''
\href{http://arxiv.org/abs/1611.00380}{{\ttfamily arXiv:1611.00380 [hep-th]}}.

\bibitem{Brunner:2016nyk}
F.~Brünner, D.~Regalado, and V.~P. Spiridonov, ``{Supersymmetric Casimir Energy
  and $SL(3,\mathbb{Z})$ Transformations},''
\href{http://arxiv.org/abs/1611.03831}{{\ttfamily arXiv:1611.03831 [hep-th]}}.

\bibitem{Shaghoulian:2015kta}
E.~Shaghoulian, ``{Modular forms and a generalized Cardy formula in higher
  dimensions},'' \href{http://dx.doi.org/10.1103/PhysRevD.93.126005}{{\em Phys.
  Rev.} {\bfseries D93} no.~12, (2016) 126005},
\href{http://arxiv.org/abs/1508.02728}{{\ttfamily arXiv:1508.02728 [hep-th]}}.

\bibitem{Shaghoulian:2015lcn}
E.~Shaghoulian, ``{Black hole microstates in AdS},''
  \href{http://dx.doi.org/10.1103/PhysRevD.94.104044}{{\em Phys. Rev.}
  {\bfseries D94} no.~10, (2016) 104044},
\href{http://arxiv.org/abs/1512.06855}{{\ttfamily arXiv:1512.06855 [hep-th]}}.

\bibitem{Klebanov:1998hh}
I.~R. Klebanov and E.~Witten, ``{Superconformal field theory on three-branes at
  a Calabi-Yau singularity},''
  \href{http://dx.doi.org/10.1016/S0550-3213(98)00654-3}{{\em Nucl. Phys.}
  {\bfseries B536} (1998) 199--218},
\href{http://arxiv.org/abs/hep-th/9807080}{{\ttfamily arXiv:hep-th/9807080
  [hep-th]}}.

\bibitem{Cardy:1986ie}
J.~L. Cardy, ``{Operator Content of Two-Dimensional Conformally Invariant
  Theories},''
\href{http://dx.doi.org/10.1016/0550-3213(86)90552-3}{{\em Nucl. Phys.}
  {\bfseries B270} (1986) 186--204}.

\bibitem{Benvenuti:2006xg}
S.~Benvenuti, L.~A. Pando~Zayas, and Y.~Tachikawa, ``{Triangle anomalies from
  Einstein manifolds},''
  \href{http://dx.doi.org/10.4310/ATMP.2006.v10.n3.a4}{{\em Adv. Theor. Math.
  Phys.} {\bfseries 10} no.~3, (2006) 395--432},
\href{http://arxiv.org/abs/hep-th/0601054}{{\ttfamily arXiv:hep-th/0601054
  [hep-th]}}.

\bibitem{Moore:1997dj}
G.~W. Moore, N.~Nekrasov, and S.~Shatashvili, ``{Integrating over Higgs
  branches},'' \href{http://dx.doi.org/10.1007/PL00005525}{{\em Commun. Math.
  Phys.} {\bfseries 209} (2000) 97--121},
\href{http://arxiv.org/abs/hep-th/9712241}{{\ttfamily arXiv:hep-th/9712241
  [hep-th]}}.

\bibitem{Gerasimov:2006zt}
A.~A. Gerasimov and S.~L. Shatashvili, ``{Higgs Bundles, Gauge Theories and
  Quantum Groups},'' \href{http://dx.doi.org/10.1007/s00220-007-0369-1}{{\em
  Commun. Math. Phys.} {\bfseries 277} (2008) 323--367},
\href{http://arxiv.org/abs/hep-th/0609024}{{\ttfamily arXiv:hep-th/0609024
  [hep-th]}}.

\bibitem{Gerasimov:2007ap}
A.~A. Gerasimov and S.~L. Shatashvili, ``{Two-dimensional gauge theories and
  quantum integrable systems},'' in {\em {Proceedings of Symposia in Pure
  Mathematics, May 25-29 2007, University of Augsburg, Germany}}.
\newblock 2007.
\newblock
\href{http://arxiv.org/abs/0711.1472}{{\ttfamily arXiv:0711.1472 [hep-th]}}.
\newblock

\bibitem{Nekrasov:2009ui}
N.~A. Nekrasov and S.~L. Shatashvili, ``{Quantum integrability and
  supersymmetric vacua},'' \href{http://dx.doi.org/10.1143/PTPS.177.105}{{\em
  Prog. Theor. Phys. Suppl.} {\bfseries 177} (2009) 105--119},
\href{http://arxiv.org/abs/0901.4748}{{\ttfamily arXiv:0901.4748 [hep-th]}}.

\bibitem{Nekrasov:2009rc}
N.~A. Nekrasov and S.~L. Shatashvili, ``{Quantization of Integrable Systems and
  Four Dimensional Gauge Theories},'' in {\em {Proceedings, 16th International
  Congress on Mathematical Physics (ICMP09): Prague, Czech Republic, August
  3-8, 2009}}, pp.~265--289.
\newblock 2009.
\newblock
\href{http://arxiv.org/abs/0908.4052}{{\ttfamily arXiv:0908.4052 [hep-th]}}.
\newblock

\bibitem{Beem:2012mb}
C.~Beem, T.~Dimofte, and S.~Pasquetti, ``{Holomorphic Blocks in Three
  Dimensions},'' \href{http://dx.doi.org/10.1007/JHEP12(2014)177}{{\em JHEP}
  {\bfseries 12} (2014) 177},
\href{http://arxiv.org/abs/1211.1986}{{\ttfamily arXiv:1211.1986 [hep-th]}}.

\bibitem{Nieri:2015yia}
F.~Nieri and S.~Pasquetti, ``{Factorisation and holomorphic blocks in 4d},''
  \href{http://dx.doi.org/10.1007/JHEP11(2015)155}{{\em JHEP} {\bfseries 11}
  (2015) 155},
\href{http://arxiv.org/abs/1507.00261}{{\ttfamily arXiv:1507.00261 [hep-th]}}.

\bibitem{Gutowski:2004ez}
J.~B. Gutowski and H.~S. Reall, ``{Supersymmetric AdS(5) black holes},''
  \href{http://dx.doi.org/10.1088/1126-6708/2004/02/006}{{\em JHEP} {\bfseries
  02} (2004) 006},
\href{http://arxiv.org/abs/hep-th/0401042}{{\ttfamily arXiv:hep-th/0401042
  [hep-th]}}.

\bibitem{Gutowski:2004yv}
J.~B. Gutowski and H.~S. Reall, ``{General supersymmetric AdS(5) black
  holes},'' \href{http://dx.doi.org/10.1088/1126-6708/2004/04/048}{{\em JHEP}
  {\bfseries 04} (2004) 048},
\href{http://arxiv.org/abs/hep-th/0401129}{{\ttfamily arXiv:hep-th/0401129
  [hep-th]}}.

\bibitem{Kunduri:2006ek}
H.~K. Kunduri, J.~Lucietti, and H.~S. Reall, ``{Supersymmetric multi-charge
  AdS(5) black holes},''
  \href{http://dx.doi.org/10.1088/1126-6708/2006/04/036}{{\em JHEP} {\bfseries
  04} (2006) 036},
\href{http://arxiv.org/abs/hep-th/0601156}{{\ttfamily arXiv:hep-th/0601156
  [hep-th]}}.

\bibitem{Chong:2005hr}
Z.~W. Chong, M.~Cvetic, H.~Lu, and C.~N. Pope, ``{General non-extremal rotating
  black holes in minimal five-dimensional gauged supergravity},''
  \href{http://dx.doi.org/10.1103/PhysRevLett.95.161301}{{\em Phys. Rev. Lett.}
  {\bfseries 95} (2005) 161301},
\href{http://arxiv.org/abs/hep-th/0506029}{{\ttfamily arXiv:hep-th/0506029
  [hep-th]}}.

\end{thebibliography}\endgroup

\end{document}